\documentclass[11pt]{article}

	\newcommand{\blind}{0}
	
    \makeatletter
    \renewcommand\section{\@startsection {section}{1}{\z@}%
                                       {-3.5ex \@plus -1ex \@minus -.2ex}%
                                       {2.3ex \@plus.2ex}%
                                       {\normalfont\fontfamily{phv}\fontsize{16}{19}\bfseries}}
    \renewcommand\subsection{\@startsection{subsection}{2}{\z@}%
                                         {-3.25ex\@plus -1ex \@minus -.2ex}%
                                         {1.5ex \@plus .2ex}%
                                         {\normalfont\itshape\fontfamily{phv}\fontsize{14}{17}\bfseries}}
    \renewcommand\subsubsection{\@startsection{subsubsection}{3}{\z@}%
                                        {-3.25ex\@plus -1ex \@minus -.2ex}%
                                         {1.5ex \@plus .2ex}%
                                         {\normalfont\itshape\normalsize\fontfamily{phv}\fontsize{14}{17}\selectfont}}
    \makeatother
	
	\usepackage{amsmath}
	\usepackage{graphicx}
	\usepackage{enumerate}
	\usepackage{xcolor}
	\usepackage{natbib} 
	\usepackage{url} 
	
            \usepackage{amsmath}
            \usepackage{amssymb}
            \usepackage{amsthm}
            \usepackage{bm} 
            \newtheorem{theorem}{Theorem}
            \newtheorem{proposition}{Proposition}
            \newtheorem{definition}{Definition}
            
            \usepackage{algorithm}
            \usepackage{algpseudocode}
            \usepackage{graphicx}
            \usepackage{wrapfig} 
            \usepackage{caption}
            \usepackage{subcaption}
            \usepackage{adjustbox}
            \usepackage{pgfplots}
            \usepackage{tikz}
            \usetikzlibrary{shapes, arrows.meta, positioning, fit, chains, matrix}

            \usepackage{titlesec}
            
            \tikzstyle{process} =[rectangle, node distance=0.8cm, text width=21em, minimum height=2em, text centered, rounded corners, draw=black, fill=white!15]
            \tikzstyle{decision}=[rectangle, node distance=0.8cm, text width=21em, minimum height=2em, text centered, rounded corners, draw=black, fill=white!20]
            \tikzstyle{algodecision}=[diamond, node distance=0.8cm, text width=21em, minimum height=2em, text centered, rounded corners, draw=black, fill=gray!20]
            \tikzstyle{arrow}=[thick, ->, >=Stealth]

            \usepgfplotslibrary{dateplot}
            \pgfdeclarelayer{foreground}
            \pgfsetlayers{main,foreground}
            \usepackage{csvsimple}
            \usepackage{subcaption}
            \usepackage{multirow}
            \usepackage{booktabs}
            \usepackage{tabularx} 
            \usepackage{colortbl}
            \usetikzlibrary{matrix}
            \usepackage{float}
            \usepackage{placeins}
            \usepackage{tcolorbox}
            \newcolumntype{C}[1]{>{\centering\arraybackslash}m{#1}}
            
            \newtcolorbox{myblock}[1]{
            	colback=blue!5!white,
            	colframe=blue!75!black,
            	fonttitle=\bfseries,
            	title=#1
            }
            \usepackage{lmodern}  
            \usepackage{anyfontsize} 
	
\begin{document}
		\def\spacingset#1{\renewcommand{\baselinestretch}%
			{#1}\small\normalsize} \spacingset{1}
		
		\if0\blind
		{
			\title{\bf Multi-objective Bayesian optimization for blocking in extreme value analysis and its application in additive manufacturing}
    
		\author{\textbf{Shehzaib Irfan$^a$, Nabeel Ahmad$^{b,c}$, Alexander                           Vinel$^a$, Daniel F. Silva$^a$,} \\
            \textbf{Shuai Shao$^{b,c}$, Nima Shamsaei$^{b,c}$, Jia Liu$^{d,\footnote{Corresponding author (jliu@ise.ufl.edu)}}$}  \\
			$^a$ Department of Industrial and Systems Engineering, Auburn University, USA,\\
            $^b$ Department of Mechanical Engineering, Auburn University, USA,\\
            $^c$ National Center of Additive Manufacturing Excellence (NCAME), Auburn University, USA,\\
            $^d$ Department of Industrial and Systems Engineering, University of Florida, USA.}
			\date{}
			\maketitle
		} \fi
		
		\if1\blind
		{

            \title{\bf \emph{Multi-objective Bayesian optimization for blocking in extreme value analysis and its application in additive manufacturing}}
			\author{Author information is purposely removed for double-blind review}
			
                \bigskip
			\bigskip
			\bigskip
			\begin{center}
				{\LARGE\bf Multi-objective Bayesian optimization for blocking in extreme value analysis and its application in additive manufacturing}
			\end{center}
			\medskip
		} \fi

	\begin{abstract}
Extreme value theory (EVT) is well suited to model extreme events, such as floods, heatwaves, or mechanical failures, which is required for reliability assessment of systems across multiple domains for risk management and loss prevention. The block maxima (BM) method, a particular approach within EVT, starts by dividing the historical observations into blocks. Then the sample of the maxima for each block can be shown, under some assumptions, to converge to a known class of distributions, which can then be used for analysis. The question of automatic (i.e., without explicit expert input) selection of the block size remains an open challenge. This work proposes a novel Bayesian framework, namely, multi-objective Bayesian optimization (MOBO-$\mathcal{D}^*$), to optimize BM blocking for accurate modeling and prediction of extremes in EVT. MOBO-$\mathcal{D}^*$ formulates two objectives: goodness-of-fit of the distribution of extreme events and the accurate prediction of extreme events to construct an estimated Pareto front for optimal blocking choices. The efficacy of the proposed framework is illustrated by applying it to a real-world case study from the domain of additive manufacturing as well as a synthetic dataset. MOBO-$\mathcal{D}^*$ outperforms a number of benchmarks and can be naturally extended to high-dimensional cases. The computational experiments show that it can be a promising approach in applications that require repeated automated block size selection, such as optimization or analysis of many datasets at once.
    \end{abstract}
			
	\noindent%
	{\it Keywords:}  Extreme value theory; Gumbel distribution; Multi-objective Bayesian optimization; Pareto front; Additive manufacturing.

	\spacingset{1.5} 
\sloppy
\section{Introduction}\label{sec:intro}
Extreme events pose great challenges to the reliability of different engineering systems, such as flash flooding, sudden wind gust affecting wind turbines or airplane wings, extreme corrosion pits rendering metallic parts unusable, or deep irregular notches on the surface of a component causing premature fatigue crack initiation. Such extreme events, despite their rarity, can be disastrous, resulting in substantial losses. With the increasing frequency of such events \citep{shen2016impact, najibi2018recent}, it is important to model and predict them to mitigate the associated risk. Extreme value theory (EVT) is a branch of statistics that is well suited to model extreme events by estimating the magnitude and probability of future outcomes based on observed history \citep{haan2006extreme}. For instance, EVT has been effectively used in engineering design to model rare but plausible extreme loads, thus giving rise to the need for robust and reliable systems \citep{clifton2008bayesian, li2014extreme, zhao2022performance, haviland1964engineering}. The application of EVT is not limited to engineering systems, but can be extended to other application areas, such as modeling of weather events \citep{engeland2004practical, coles1998extreme, shen2016impact}.

EVT relies primarily on two sampling techniques, i.e., block maxima (BM) \citep{gumbel1958statistics} and peaks-over-threshold (POT) \citep{pickands1975statistical}, to systematically select large values in the domain of the possible occurrence of extreme events to estimate their distribution. BM divides the domain under observation into equal-sized and non-overlapping blocks and picks the maximal value. POT sets a (relatively large) threshold and considers only outcomes exceeding it. Neither technique is perfect \citep{bucher2021horse}: picking only the maximum value from each predefined block, BM may ignore other large outcomes in the same block, while POT may pick all large values from only a small region in the domain, biasing the analysis. Preference for BM or POT depends primarily on the context of the application and the type of observed data.

Traditionally, BM uses application-specific inherent blocking. One popular approach is calendar blocking (e.g., daily, weekly, etc.) for applications such as finance \citep{bucher2018inference}, hydrology \citep{asadi2018optimal}, rainfall \citep{gurung2021modelling}. However, owing to the recent data explosion in many modern data-intensive applications such as high-resolution manufacturing scans \citep{ahmad2025determining, nikfar2025extreme}, these applications need extreme event analysis but do not have natural blocking structures for their datasets. Therefore, there has been a renewed interest in the question of \textit{how should the block size be selected in the block maxima (BM) analysis} \citep{zou2021multiple}. In the existing literature, there is a lack of a systematic way to answer this question, particularly, if we focus on approaches to optimize the blocking in BM automatically (i.e., without a human expert) \citep{bucher2021horse}. This research seeks to answer this important question by formulating the underlying problem as a generalizable multi-objective optimization problem to find the optimal blocking (i.e., block dimensions) that balances the accuracy of predicting extreme events and the goodness-of-fit of the distribution of extreme events in EVT, which can be confidently used in different applications, therefore, making the reliability analysis more robust.

Recently, some research efforts have considered methodologies for block size selection that could be used to manually select different block sizes and assess their performance based on different metrics. For example, \cite{wang2018} proposed a methodology that uses an entropy-based indicator that makes use of different goodness-of-fit criteria, including Kolmogorov-Smirnov and Chi-square, to select the block size with the smallest indicator value. This method focuses on the goodness-of-fit of the distribution of the extreme values but neglects the accuracy of using the fitted distribution to estimate the extreme values in EVT. Further, they select the block size options beforehand and then perform the selection from that list only, therefore making this procedure somewhat arbitrary and less suitable high-dimensional datasets. 
\cite{ozari2019} proposed a parametric approach to manually select the block size based on the goodness-of-fit of the data and the error in the predicted quantile. The manual nature of the proposed approach makes it very tedious and sometimes impractical for application areas where the selection range of the block size may be large.
\cite{Dkengne2020} presented an algorithm that can be used to select the appropriate block size for generalized extreme value (GEV) distribution. This algorithm consists of three stages. First, a list of pre-selected block sizes is used to fit the GEV distribution. Second, the smallest size that simultaneously stabilizes the distribution parameters is picked. Third, the block size with the largest shape parameter, significantly similar to the shape parameter of the stabilizing block size, is selected. However, the choice of candidate sizes is again somewhat arbitrary, and much of the process requires expert feedback and therefore is difficult to automate. Finally, the search is essentially an enumerative procedure, which could be  intractable if the search space is large.

Another related work, \cite{bucher2018inference}, considers overlapping blocks. The authors propose the idea of a sliding block of constant size over the domain and recording the maxima from that as a stationary time series. While this violates the independently and identically distributed (i.i.d.) assumption on the maxima, it remains stationary, and the approach has been shown to significantly reduce the asymptotic variance for the maximum likelihood estimates (MLE) of the parameters. Authors have also shown that MLE estimates become stable as a function of block size if sliding blocks are used. See also \cite{zou2021multiple} for another approach to this idea. Note, though, that the question of optimal block size selection remains valid for this setting as well. 

In this work, we formulate the block size selection problem as a bi-objective optimization problem and propose a Bayesian framework to solve it by constructing a set of non-dominated solutions, which is an approximation of the actual Pareto front. The two objectives are the accuracy of prediction and the goodness of the fit of data to the theoretically suitable distribution of the extremes, Gumbel in this case. Naturally, the former is more directly useful, since it is usually the goal of the overall analysis. 
The rationale behind using goodness of fit as a second objective is to prevent overfitting the training data if the prediction accuracy is used as the only objective. In this sense, it can be seen as a regularization technique, making sure that the accuracy is achieved by capturing a real trend in the data (close to the true distribution of extremes) rather than random noise. 
The case study in Section \ref{sec:case studies} supports this decision and highlights a situation where the block size obtained from optimizing solely for prediction accuracy performed worse on the testing dataset, compared to other non-dominated solutions identified. Furthermore, including goodness of fit aligns the proposed methodology with the literature, where that is often the preferred objective function \citep{wang2018,ozari2019}. 

The proposed framework provides a general systematic way for selecting block size in extreme value analysis based on BM. It advances current practices of inherent blocking or manual selection to a systematic algorithm to explore the optimal solutions in the feasible region, even in high dimensions. We demonstrate the performance of the framework using a real-world case study and a simulated dataset. We show that our proposed framework offers a very good balance between the quality of the solution and the computational cost against the benchmark approaches. Furthermore, the framework's computational costs can be expected to scale well with data dimensionality, and it, by design, does not require expert input.   As a result, it is well-suited for applications that involve a repeated analysis of many datasets, for example, optimization with EVT estimate as an objective or a constraint.

The rest of the paper is structured as follows. The proposed multi-objective Bayesian optimization framework is detailed in Section \ref{sec:method}. Its efficacy is demonstrated and validated using a real case study from manufacturing and a simulated case study in a high-dimensional domain in Sections \ref{sec:case studies} and \ref{sec:simulation}. The computational burden of the proposed method is studied in Section  \ref{sec:time_comparison}. Finally, conclusions and future work are described in Section \ref{sec:conclusion}. 

\section{Proposed methodology}\label{sec:method}
The proposed multi-objective Bayesian optimization (MOBO-$\mathcal{D}^*$) framework harnesses the power and flexibility of Bayesian statistics and Bayesian optimization to find the optimal block size in BM for EVT. Using Bayesian optimization, it assigns a random process (Gaussian process) as a surrogate model for the objective functions and then uses the 
maximization of an acquisition function to explore the feasible region to approximate the Pareto front in an efficient manner. Figure \ref{fig:flowchart} gives an overall summary of the proposed framework. In an extreme value analysis, raw data are acquired from the domain of interest, and the procedure of optimizing block size in BM for EVT include:  

\begin{figure}[H]
	\centering
\includegraphics[height=0.55\linewidth]{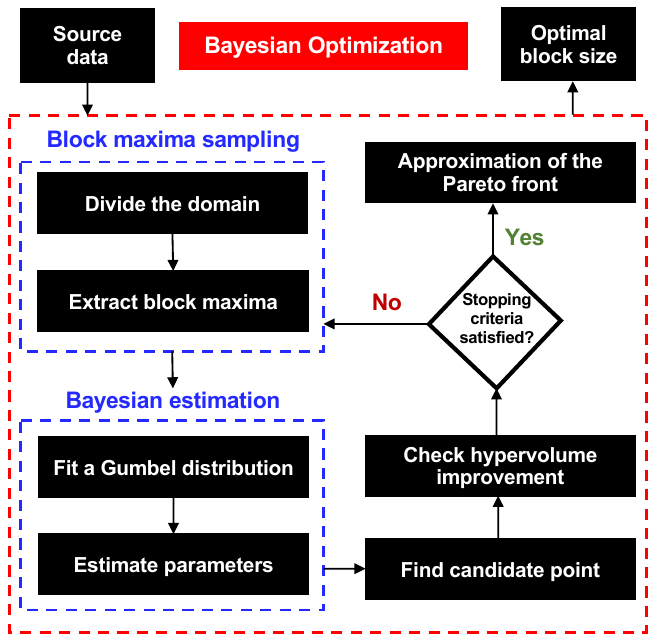}
	\caption{Flowchart describing the proposed methodology to select the optimal block size.}
	\label{fig:flowchart}
\end{figure}

\begin{enumerate}
  \item Block maxima sampling: divide the domain of interest with some block sizes for initialization and extract block maxima values.
  
  \item Bayesian estimation: estimate the parameters of the fitted generalized extreme value (GEV) distribution (here Gumbel distribution) from the block maxima values. 
  
  \item Bayesian optimization: formulate a bi-objective optimization problem to find the optimal block size to satisfy the goodness of the fit of data to the estimated distribution and the accurate prediction of extreme values. The potential optimal solutions are obtained by using hyper-volume improvement in an approximated Pareto front. 
\end{enumerate}

The process is repeated until the stopping criteria are satisfied, such as the convergence of results or the exhaustion of the number of iterations. The approximation of the Pareto front is obtained based on the optimization results once the iterations are complete. The optimal block size is selected based on the decision maker's preference among the non-dominated solutions in the approximated Pareto front.  

\subsection{Background: extreme value theory}\label{sec:prelim_results}
This subsection summarizes the theoretical foundation of the block maxima method.
The following result due to \citet{fisher1928limiting} forms the basis of the approach. Let a random variable $M_k =\text{max}\{Y_1, Y_2, \dots , Y_k\}$, where $Y_1, Y_2, \dots, Y_k$ is a sample of random variable $Y$  with finite variance.

\begin{theorem} 
	If there exists a sequence of real numbers  $a_k>0$ and $b_k \in \mathbb{R}$ such that the following limit converges to a non-degenerate distribution function:
	\begin{align} \label{eq:FT_thm}
		\lim_{k\rightarrow \infty} \mathbb{P} \left\{\frac{M_k - b_k}{a_k} \leq x \right\} = G(x),
	\end{align}
	then $G(x)$ belongs to the Generalized Extreme Value (GEV) distribution family, given as
	\begin{align} \label{eq:GEV}
		G(x) = \exp \left\{- \left[1 + \gamma \left\{\frac{x -\mu}{\sigma}\right\}\right]^{-\frac{1}{\gamma}}\right\},
	\end{align}
	where $\left\{x \in \mathbb{R} : 1 + \gamma\left(\frac{x-\mu}{\sigma}\right) > 0 \right\}$, $\gamma, \mu \in \mathbb{R}, \: \sigma >0$. Parameter $\mu$ is the location parameter, $\sigma$ is the scale, and $\gamma$ is the shape.
\end{theorem} 

The GEV family includes three distributions: Fr\'{e}chet ($\gamma>0$), Weibull ($\gamma < 0$), and Gumbel ($\gamma  = 0$). From (\ref{eq:GEV}) it can be seen that Fr\'{e}chet has a polynomial decay in the tail, Weibull has a finite right endpoint, while Gumbel has an exponential decay. Therefore, the Gumbel distribution becomes important and relevant in applications where extremes are unlikely to be very large, and an upper bound is likely to exist but cannot be known precisely. This distribution has been frequently used to model the extremes in different applications like engineering \citep{niemann2013statistics, melchers2021new}, climatology \citep{kang2015determination}, and finance \citep{purohit2022european}. We use the Gumbel distribution to model the extreme values in our proposed framework. 

Using $\gamma = 0$ and the following limit:
\begin{align} 
	\lim\limits_{\gamma \to 0} \left[ 1 + \gamma \Biggl\{  \left(- \frac{x-\mu}{\sigma}\right)\Biggr\}\right] ^{-\frac{1}{\gamma}} = \exp\Biggl\{ \left(- \frac{x-\mu}{\sigma}\right)\Biggr\},
	\end{align}
in (\ref{eq:GEV}), we can write the cumulative distribution and the density functions of Gumbel distribution as follows:
\begin{align}\label{eq:gumbel_cdf}
	G(x) = \exp \left\{-\exp \left\{-\left(\frac{x-\mu}{\sigma}\right)\right\}\right\};
\end{align} 
\begin{align}\label{eq:gumbel_pdf}
	g(x) = \frac{1}{\sigma} \exp\left\{-\left(\frac{ x-\mu}{\sigma}\right) -  \exp \left\{- \left(\frac{ x-\mu}{\sigma}\right) \right\}       \right\}.
\end{align} 

\subsection{Problem formulation}\label{sec:prob form}
This subsection formulates the problem in terms of two objective functions for the MOBO-$\mathcal{D}^*$ framework. Suppose that we observe a phenomenon in an $n$-dimensional domain, and the ultimate goal is to model and gain insights about the extreme realizations of it with EVT. We divide the domain into blocks and pick the block maxima from each block to fit a distribution of the extreme event in this domain, following the theoretical results described above. 

The question addressed here is the size of the blocks considered. We define the set of block size parameters $\mathcal{D} =\{D_1, D_2, \dots, D_n\} \in \mathbb{Z}^n$ for the domain such that $D_j$ divides the $j^{th}$ dimension into $D_j$ parts, with user-defined upper bound $U_j$ on these block size parameters: $1 \leq D_j \leq U_j < + \: \infty \;, \forall \; j = 1, \dots, n$. Then, a block maxima sample $x=\{x_1, x_2, \dots, x_{m(\mathcal{D})}\}$ is obtained and used to estimate the parameters of the Gumbel distribution $\mathbf{\Theta} = (\mu,\sigma)$. Here, $m(\mathcal{D})$ is the number of block maxima values and depends on the set of block size parameters $\mathcal{D}$, i.e., $m(\mathcal{D}) = \prod_j D_j$. 
The Gumbel distribution can be used to estimate the magnitude of an extreme event $\hat{q}$ with number of blocks equal to $m(\mathcal{D})$ from (\ref{eq:gumbel_cdf}) as 
\begin{align}
	 \mathbb{P}\{X > \hat{q}\} &:= 1 - G (\hat{q}) = \frac{1}{m(\mathcal{D})},   \\
	\Rightarrow  \; \hat{q} &= \hat{\mu} - \hat{\sigma} \log \big[\log(m(\mathcal{D})) -\log(m(\mathcal{D})-1)\big]. \label{eq:return_level}
\end{align}

The prediction error can then be calculated as follows and will be used as the first objective in the optimization:
\begin{align}\label{eq:prediction_accuracy}
	\left|  \frac{q-\hat{q}}{q} \right|,
\end{align}  
where $q$ is the observed most extreme value, i.e., $q = \max_i\{ x_i\} $, and $\hat{q}$ is given above \eqref{eq:return_level}.

The Kolmogorov-Smirnov (KS) test statistic, which quantifies the goodness-of-fit (GoF) of the estimated parameters, given in (\ref{eq:KS_test_statistic}), will be used as the second objective in the optimization. We use the KS statistic due to its non-parametric nature, simplicity, robust detection of distributional differences, and low sensitivity to outliers.
\begin{align}\label{eq:KS_test_statistic}
	\underset{x}{\text{sup}} \left| G_{\bm{\mathcal{D}}}(x) - G_{EDF}(x) \right|,
\end{align}
where $G_{\bm{\mathcal{D}}}(x)$ is the cumulative distribution function (CDF) of the fitted Gumbel distribution and $G_{EDF}(x)$ is the empirical distribution function (EDF) based on the sample $x=\{x_1, x_2, \dots, x_{m(\mathcal{D})}\}$. The smaller the difference between the two, the better the distribution fits the data.

Both (\ref{eq:prediction_accuracy}) and (\ref{eq:KS_test_statistic}) are to be minimized to find the set of optimal block size parameters $\mathcal{D}^{*}$. The goodness of fit is important because it ensures that the proposed methodology is robust to overfitting. For instance, if the fit is bad but the accuracy is high, it may indicate overfitting the data to the model, which could eventually lead to bad out-of-sample performance. It is demonstrated in our case study (Section \ref{sec:case_study_AM_d*_one_param}), where $\mathcal{D}$, which showed the highest accuracy during the training process, performed poorly on the unseen data. Therefore, we use both objectives in our proposed framework. A similar approach has been adopted by \cite{wang2018}, \cite{ozari2019}, and \cite{Dkengne2020}, but in different contexts. 

Our proposed framework uses point estimates for the Gumbel distribution parameters in both objectives and is independent of the choice of the parameter estimation methodology. Commonly used estimators for the point estimation of distribution parameters are maximum likelihood (MLE) and probability weighted moment (PWM) estimators \citep{ferreira2015block}. On the other hand, 
Bayesian inference is independent of regularity assumptions, robust against overfitting, and is capable of characterization of the uncertainty associated with these parameters, which led to interest in using Bayesian inference for this goal \citep{moins2023reparameterization}. Therefore, we have opted to use Bayesian inference for estimating the parameters and give details of the procedure in the \ref{app:parameter_estimation}.

As mentioned in Section \ref{sec:prob form}, the aim here is to minimize two objective functions to find the set of optimal block size parameters $\mathcal{D}^* =\{D_1^*, D_2^*, \dots, D_n^*\}$. The resulting problem is a bi-objective mathematical program in an $n$-dimensional domain and can be written as
\begin{align}\label{eq:optimizatin_problem}
	\mathbf{min}& \hspace{1em} f_1(\mathcal{D}) = \left|  \frac{q-\hat{q}}{q} \right|, \\
	\mathbf{min}& \hspace{1em} f_2(\mathcal{D}) = \underset{x}{\text{sup}} \left| G_{\bm{\mathcal{D}}}(x) - G_{EDF}(x) \right|, \\
        \mathbf{s.t.}& \hspace{1em} 1 \leq D_j \leq U_j \;\;  \;\; j = 1, \dots, n, \\
	& \hspace{1em} D_j \in \mathbb{Z}_+ \hspace{2em} \;\;     j = 1, \dots, n.
\end{align}
Both objectives are nonlinear and non-convex, and in fact, no closed-form expression connecting the decision variables with the objectives is available. Further, estimation of the objectives requires a separate optimization step (fitting of the EVT parameters). As a result, we do not expect that off-the-shelf global optimization solvers will be suitable for this task, and instead propose a custom optimization approach.

\subsection{Bayesian optimization for blocking in EVT}\label{sec:bayesian_optimization}
We propose a Multi-Objective Bayesian Optimization algorithm (MOBO-$\mathcal{D^*}$) to solve the above optimization problem.  Bayesian optimization places a probabilistic \textit{surrogate} model on the objectives individually, with the most popular choice being Gaussian processes \citep{gnanasambandam2024deep} and then a multi-objective \textit{acquisition} function uses that surrogate model to decide which points to evaluate next based on the expected improvement in the optimization direction \citep{roussel2021multiobjective}. 

In MOBO-$\mathcal{D^*}$, we assign Gaussian Processes (\textit{GP}) as surrogate models for both objective functions: $f(\mathcal{D}) \sim GP(\eta(\mathcal{D}), \zeta(\mathcal{D}, \mathcal{D}'))$, where $\eta$ and $\zeta$ are mean and co-variance function respectively, and $\mathcal{D}$ and $\mathcal{D'}$ represent two different sets of values of $\mathcal{D}$. Generally, a constant value (e.g., 0) is assumed as a prior for the mean $\eta$, and radial basis functions, which encode higher correlations between closer variables, is used as prior for co-variance function $\zeta$. To enable MOBO-$\mathcal{D^*}$, we give the following definitions before presenting the steps of our algorithm.

\begin{definition}\label{def:pareto_set}
	For minimization of M objectives, a solution $\mathcal{D}$ \textbf{dominates} another solution $\mathcal{D}'$, i.e. $\mathcal{D} \succ \mathcal{D}'$, if and only if $f_{m}(\mathcal{D}) \leq f_{m}(\mathcal{D}') \; \forall \; m = 1,\dots, M$ and there exists $ m' \in \: \{1,\dots,M\}$ such that $f_{m'}(\mathcal{D}) < f_{m'}(\mathcal{D}')$. Then a \textbf{non-dominated (ND) solution set} $\mathcal{P}^{sol}$ can be defined as $\{ \mathcal{D} \; s.t. \; \nexists \; \mathcal{D}' : \bm{f}(\mathcal{D}') \succ \bm{f}(\mathcal{D})\}$ and a \textbf{non-dominated (ND) set} $\mathcal{P}$ can be defined as $\{\bm{f}(\mathcal{D}) \in \mathbb{R}^M : \mathcal{D} \in \mathcal{P}^{sol}\}$, where $\bm{f}(\mathcal{D})= (f_1(\mathcal{D}), f_2(\mathcal{D}), \dots, f_M(\mathcal{D}))$.
\end{definition}

\begin{definition}\label{def:hypervolume_indicator}
	Given a reference point $\bm{r} \in \mathbb{R}^M$, the hyper-volume (HV) of the ND set $\mathcal{P}$ is the M-dimensional Lebesgue measure $\lambda_M$ of the space dominated by $\mathcal{P}$ and bounded from above by $\bm{r} : HV (\mathcal{P}, \bm{r})= \lambda_M \Bigl(\bigcup_{v=1}^{|\mathcal{P}|}[\bm{r}, \mathcal{P}_v]\Bigr)$, where $\mathcal{P}_v \in \mathcal{P}$ and $[\bm{r}, \mathcal{P}_v]$ is the hyper-rectangle bounded by vertices $\bm{r}$ and $\mathcal{P}_v$.
\end{definition}
Without loss of generality, $\bm{r}$ can also bound the space dominated by $\mathcal{P}$ from below if the objective functions are to be maximized. In our case, we initialize $\bm{r}$ at the origin and propose a heuristic-based approach to find an appropriate place for it in the objective space.
\begin{definition}\label{def:hypervolume_improvement}
	Given the ND set $\mathcal{P}$ and reference point $\bm{r}$, the hyper-volume improvement (HVI) of a candidate point $\mathcal{D}_{cand}$ is $: HVI\bigl(\bm{f}(\mathcal{D}_{cand}), \mathcal{P}, \bm{r}\bigr) = HV\bigl(\mathcal{P} \cup\bm{f}(\mathcal{D}_{cand}), \bm{r}\bigr) - HV(\mathcal{P}, \bm{r})$.
\end{definition}
Figure \ref{fig:illustration_pareto_front} illustrates the reference point $\bm{r}$, which bounds the space dominated by the ND set $\mathcal{P}$ from above. The blue area highlights the hyper-volume ($HV$) bounded by the reference point and ND set. $\mathcal{D}_{cand}$ is the candidate point that gives hyper-volume improvement ($HVI$) highlighted in red.
\begin{figure}[h]
    \centering	
    \includegraphics[height=0.25\linewidth]{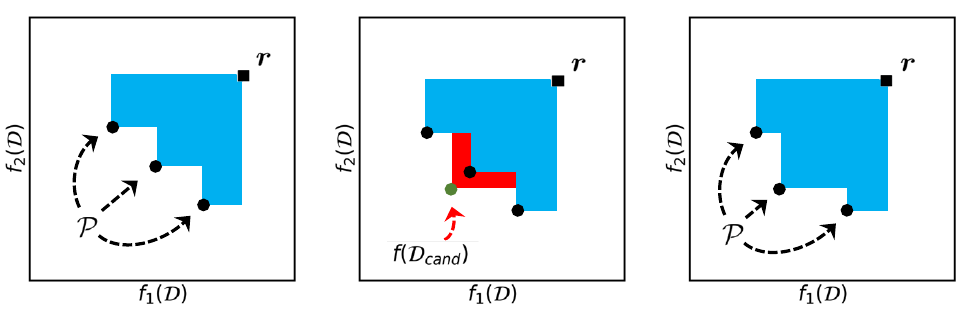}
    \spacingset{1}
    \caption{(left) The reference point $\bm{r}$ bounds the space dominated by the ND set $\mathcal{P}$. The blue area highlights the $HV$ bounded by the reference point and ND set. (center) $\mathcal{D}_{cand}$ is the candidate point that gives the $HVI$ highlighted in red. (right) After the inclusion of the candidate point, the point previously in the ND set but dominated by the candidate point is removed from the ND set.}
    \vspace{-10pt}
    \spacingset{1.5}
    \label{fig:illustration_pareto_front}
\end{figure}
\begin{definition}\label{def:acquisition_function}
	Acquisition function $\alpha_{EHVI}$, is the expectation of HVI over the joint posterior distribution of objective functions, $\Phi_{\eta, \zeta}$, such that $\alpha_{EHVI} = \mathbb{E}\Bigl[ HVI (\bm{f}(\mathcal{D}_{cand}), \mathcal{P}, \bm{r})\Bigr]= \int_{0}^{\infty}HVI(\bm{f}(\mathcal{D}_{cand},\mathcal{P},\bm{r}))  \; d\Phi_{\eta, \zeta}$.
\end{definition}
The two objective functions can be modeled using independent GPs,  and $\alpha_{EHVI}$ can be calculated using an efficient box decomposition algorithm, whose time complexity on a single-threaded machine is $O(|\mathcal{P}|)$ where $|\mathcal{P}|$ is the number of non-dominated points as of that iteration \citep{yang2019multi, daulton2020differentiable}. In order to find a good candidate solution $\mathcal{D}_{cand}$ such that $\bm{f}(\mathcal{D}_{cand})$ enters the non-dominated (ND) set $\mathcal{P}$ we need to solve the equation, $\mathcal{D}_{cand} = \underset{\mathcal{D}}{\text{argmax }}\alpha_{EHVI}$, which will contribute towards hyper-volume improvement in each step.

\spacingset{1.25}
\begin{algorithm}[h]
	\caption{Multi-Objective Bayesian Optimization for selecting optimal block size $\mathcal{D}^*$.}
	\begin{algorithmic}[1]
		\State \textbf{Input:}  raw data from the source, window size $N$, tolerance level $\epsilon$, user-defined bounds $U_j$, reference point at origin: $\bm{r} = (r_1,r_2) = (0,0)$, objective functions $\bm{f}(\mathcal{D}) = \bigl(f_1(\mathcal{D}), f_2(\mathcal{D})\bigr)$, $\mathbf{C}_{\epsilon}=2\epsilon$.
		\State Initialize with $k$ points such that $\mathcal{P}^{sol} = \bigl\{\bigr\} \Rightarrow \mathcal{P}=\bigl\{\bigr\} \Rightarrow HV(\mathcal{P},\bm{r}) \gets 0$. $t \gets 1$. 
            \While{$\mathbf{C}_{\epsilon} > \epsilon$}
		\State $\mathcal{D}_{cand} \gets \underset{\mathcal{D}}{\text{argmax}} \;\mathbb{E}\Bigl[ HVI (\bm{f}(\mathcal{D}), \mathcal{P}, \bm{r})\Bigr]$.
		\State $HV(\mathcal{P},\bm{r}) \gets \lambda_M \Bigl(\bigcup_{v=1}^{|\mathcal{P}|}[\bm{r}, \mathcal{P}_v]\Bigr)$. 
		\If{$HV(\mathcal{P},\bm{r}) = 0$}
		\State $r_{step} =  \text{min}(f_1(\mathcal{D}_{cand}), f_2(\mathcal{D}_{cand}))$.
		\State $r_1 \gets r_1 + \beta \times r_{step} $, $r_2 \gets r_2 + \beta \times r_{step} $.
		\State $t \gets 1 $, GOTO 3.
		\EndIf
		\State $HVI\bigl(\bm{f}(\mathcal{D}_{cand}), \mathcal{P}, \bm{r}\bigr) \gets HV\bigl(\mathcal{P} \cup\bm{f}(\mathcal{D}_{cand}), \bm{r}\bigr) - HV(\mathcal{P}, \bm{r})$.
		\If{$HVI\bigl(\bm{f}(\mathcal{D}_{cand}), \mathcal{P}, \bm{r}\bigr) > 0$}
        \State $\mathcal{P}^{sol} \gets \mathcal{P}^{sol} \cup \bigl\{\mathcal{D}_{cand} \bigr\}$.
		\State $\mathcal{P} \gets \mathcal{P} \cup \bigl\{\bm{f}(\mathcal{D}_{cand})\bigr\} \backslash \bigl\{\mathcal{P}_v \; : \; \bm{f}(\mathcal{D}_{cand}) \succ \mathcal{P}_v \; \forall \;  v =1, \dots, |\mathcal{P}| \bigr\}$.
		\EndIf
        \If{$t>N$}
		\State $\mathbf{C}_{\epsilon} = \frac{1}{N} \sum_{a=t-N+1}^{t}|HV_a - HV_{a-1}|$. 
            \State $t \gets t + 1$.
            \Else{}
            \State $t \gets t + 1$.
		\EndIf
	\EndWhile 
	\end{algorithmic}
	\label{algo:MOBO}
\end{algorithm}
\spacingset{1.5} 

MOBO-$\mathcal{D^*}$ (Algorithm \ref{algo:MOBO}) is initialized by placing the reference point $\bm{r}$ at the origin, while the theoretical minimum value of the objectives is zero. The hyper-volume is set to zero, and the ND set $\mathcal{P}$ is set to be empty. Then in each iteration, $\mathcal{D}_{cand}$ is found by maximizing $\alpha_{EHVI}$ and hyper-volume is calculated. As long as the hyper-volume is zero, reference point $\bm{r}$ is moved away from the origin in small steps: $\bm{r} \leftarrow \bm{r} + \beta \times r_{step}$, where $\beta \in (0,1]$ is the growth factor, and $r_{step}$ is the minimum of the two objective functions calculated at $\mathcal{D}_{cand}$, i.e., $r_{step} =  \text{min}(f_1(\mathcal{D}_{cand}), f_2(\mathcal{D}_{cand}))$. Smaller $\beta$ would explore the feasible objective space more slowly, yet more precisely, when the algorithm tries to minimize the objective with respect to $\bm{r}$. Once the reference point has been placed appropriately, the optimization proceeds with maximizing the hyper-volume and constructing the sets $\mathcal{P}$ and $\mathcal{P}^{sol}$. This process is repeated until the stopping criterion is met. We use a moving average-based criterion as in (\ref{eq:conv_criteria}) that monitors the improvement in HV over a number of iterations and stops the algorithm when it is met. We define
\begin{align}\label{eq:conv_criteria}
    \mathbf{C}_{\epsilon} = \frac{1}{N} \sum_{a=t-N+1}^{t}|HV_a - HV_{a-1}| \leq \epsilon,
\end{align}
where $N$ is the number of iterations for which we are monitoring the change in HV, $t$ is the current iteration number, and $\epsilon$ is the tolerance. Note that the user can define the values of $N$ and $\epsilon$ based on the context of the problem. Once the optimization process is complete, it gives the ND solution set $\mathcal{P}^{sol}$. Then, the decision maker, as per their preference, can pick an element from $\mathcal{P}^{sol}$, which we define as $\mathcal{D}^*$. 

\section{Case study: surface roughness of additively manufactured parts}\label{sec:case studies}
The efficacy of MOBO-$\mathcal{D^*}$ is demonstrated using a real-world case study in the field of additive manufacturing (AM). Metal AM technology has multiple advantages, including minimizing waste, costs, development time, and customization of the manufactured parts. However, fatigue performance has been a key barrier to their adoption in critical applications \citep{sanaei2021defects,li2024nondestructive}. One of the detrimental factors impacting the fatigue performance of AM parts is surface roughness, which exists in the form of peaks and valleys \citep{Murakami_2019}. These valleys are scattered across the surface, and the deepest valleys can be viewed as extreme events on the surface \citep{fox2021prediction, nikfar2025extreme}. The maximum valley depth ($S_v$) has been shown to have a good correlation with fatigue performance \citep{gockel2019influence, lee2021surface}, yet scanning the entire part to find it is time-consuming and often impractical.

Our proposed MOBO-$\mathcal{D}^*$ can be leveraged to find the optimal block size that could be used to: 1) model the surface roughness as a distribution of the depth of valleys and 2) use that distribution to make the prediction of the deepest valley in the unscanned parts for fatigue performance and reliability analysis.

\subsection{Description of surface measurements}\label{sec:data_description}
The data used here comes from a two-dimensional optical surface scan of the gauge area of an AM specimen, as illustrated in Figure \ref{fig:point_cloud}. 
Dimensions of the total scanned area are $4 \;  mm \times 5 \; mm$, and each point on the surface represents the surface height at that particular location in the 2D plane. This scanned surface has $\sim$ 200 million such points. In EVT, the surface can be divided into blocks as shown, where the dimensions of each block are indicated as $w_1 \times w_2$. Figure \ref{fig:point_cloud} also illustrates an example block from the surface where the blue color represents the peaks and the red color represents the valleys. 

\begin{figure}[H]
\centering
\includegraphics[width=0.45\textwidth]{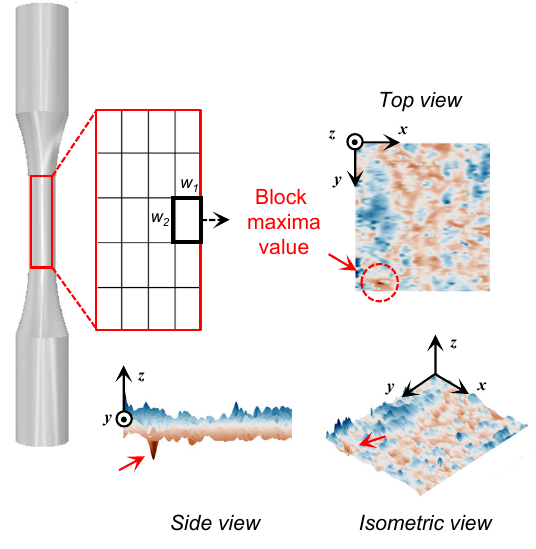}
\spacingset{1}
\caption{A specimen with the scanned surface (red rectangle). It is divided into blocks, with dimensions $w_1$ and $w_2$, in EVT to find the deepest valley of the whole part.}
\spacingset{1.5}
\label{fig:point_cloud}
\end{figure}

To find the optimal block size for extreme value analysis, we use MOBO-$\mathcal{D^*}$ to approximate the Pareto front. As MOBO-$\mathcal{D^*}$ can explore the feasible region with different dimensionality, we demonstrate the methodology in two cases: blocking with square blocks, a one-dimensional (1D) case ($n=1$); and blocking with rectangular blocks, a two-dimensional (2D) case ($n=2$).  

\subsection{MOBO-$\mathcal{D^*}$ square block size optimization (1D)}\label{sec:case_study_AM_d*_one_param}
A smaller $1 \;  mm \times 5 \; mm$ area is used as the input data for BM evaluation. The goodness of fit objective ($f_2$) and $\hat{q}$ are evaluated solely based on this sample. The deepest valley for the overall gauge area of the specimen is used as the value of $q$ in the definition of $f_1$. For the first part of the analysis, we restrict the block sizes to be equal, i.e., there is a single block size variable $\mathcal{D}=\{D_1\}$. The surface is then split into $D_1\times 5D_1$ blocks to account for the shape of the scans. We set  $U=200$.

We initialize the optimization process by placing the reference point $\bm{r}$ at the origin as observed in Figure \ref{fig:iterations_AM} (top); it is updated as described in Steps 7 and 8 in Algorithm 1 of the MOBO-$\mathcal{D^*}$. For stopping criteria, we choose $N=5$ and $\epsilon = 10^{-5}$ to allow enough exploration of the objective space. MOBO-$\mathcal{D}^*$ stops after 20 more iterations to construct the non-dominated (ND) set, as can be observed from the iterations (red to green in ascending order) in Figure \ref{fig:iterations_AM} (bottom). The same figure also shows the evolution of hyper-volume with iterations. As seen in Figure \ref{fig:iterations_AM} (bottom), we observe relatively fast convergence as indicated by $HVI$. Based on the optimization process shown in Figure \ref{fig:iterations_AM}, the Pareto front can be approximated as shown in Figure \ref{fig:pareto_front_AM} (left). 

\begin{figure}[H]
	\centering
	\includegraphics[width=0.85\linewidth]{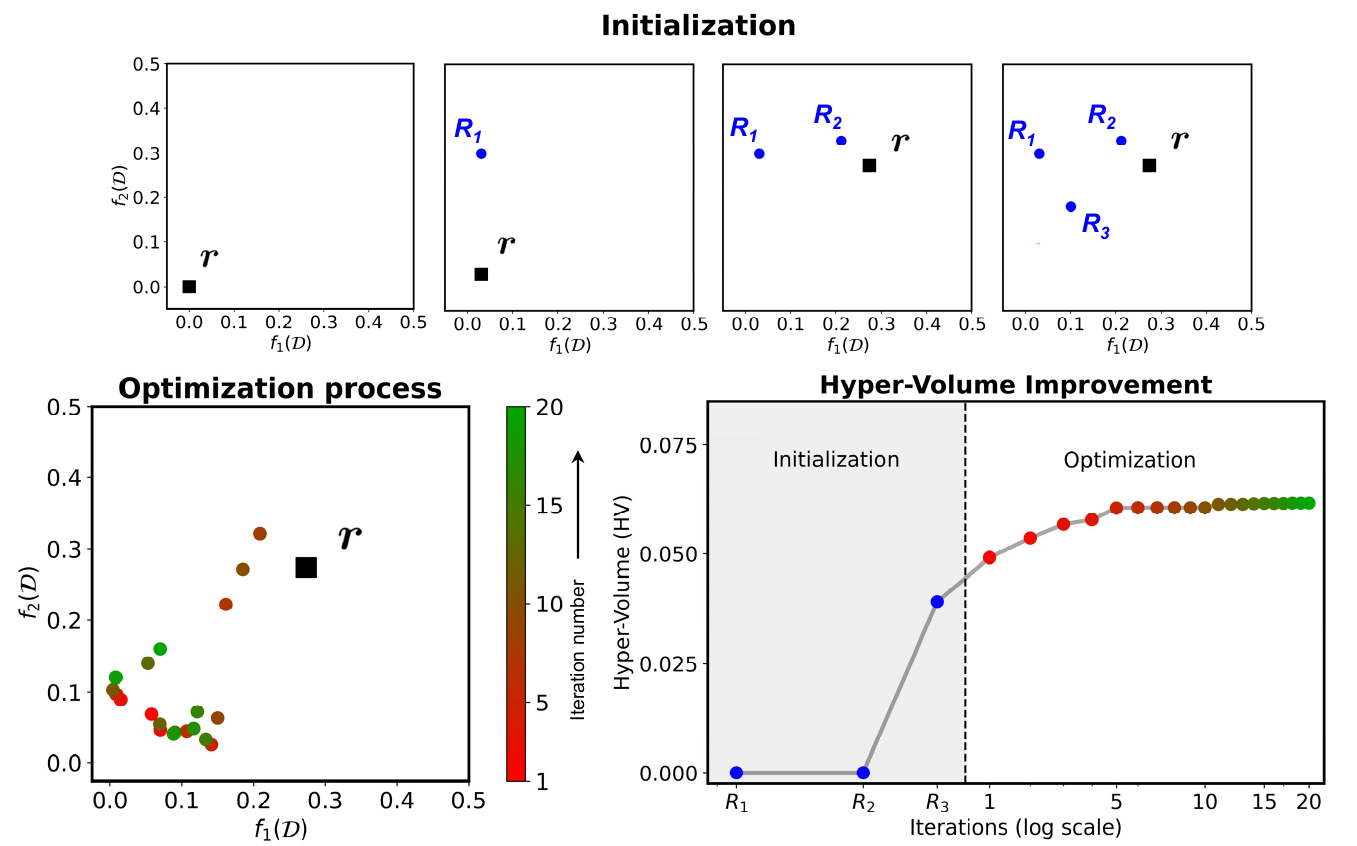}
    \spacingset{1}
    \caption{(top) Initialization is done by placing $\bm{r}$ at the origin. In the first and second steps, no gain in hyper-volume is observed, therefore $\bm{r}$ is updated. The reference point $\bm{r}$ is fixed in the third step, and the optimization process proceeds further. (bottom) The optimization process consists of 20 iterations, colored from red to green in ascending order.}
    \spacingset{1.5}
    \label{fig:iterations_AM}
\end{figure}

\begin{figure}[ht]
    \centering
    \includegraphics[height=0.35\linewidth]{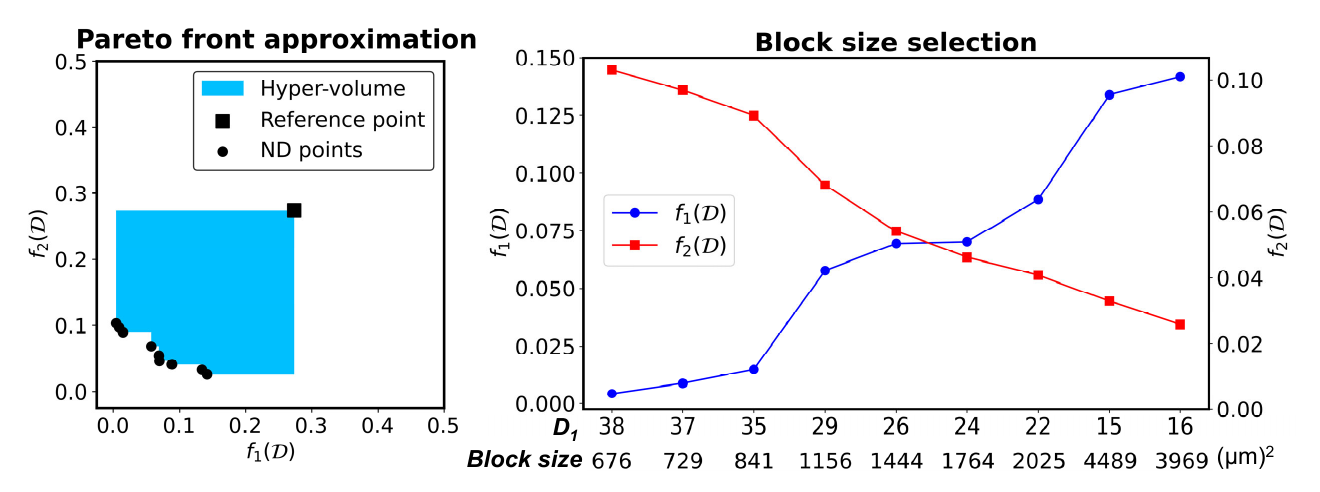}
    \spacingset{1}
    \caption{(left) Approximated Pareto front based on the optimization process. (right) Both objective functions are plotted against the block size parameter $D_1$ and show a trade-off between the two, which the decision maker can use to get the $\mathcal{D}^*$.}
    \spacingset{1.5}
    \label{fig:pareto_front_AM}
\end{figure}
\spacingset{1.5}

Table \ref{tab:Pareto_front_comparison} lists the obtained non-dominated solutions, and Figure \ref{fig:pareto_front_AM} (right) gives the resulting block size. There is a clear trade-off between the objectives, as shown in Figure \ref{fig:pareto_front_AM} (right); therefore, the choice of $\mathcal{D}^*$ depends on the decision-maker's preference. For smaller blocks, GoF is not as good as accuracy, while on the other end of the spectrum, accuracy is not as good as GoF.

\spacingset{1}
\begin{table}[h]
\centering
    \caption{Comparison between the Pareto front obtained from the enumeration and the one approximated by MOBO-$\mathcal{D}^*$.}
    \begin{tabular}{|C{0.5cm} C{1cm} C{1cm} | C{0.5cm} C{1cm} C{1cm}|}
    \hline
    \multicolumn{3}{|c|}{Enumeration}  & \multicolumn{3}{c|}{MOBO-$\mathcal{D}^*$} \\ 
    \hline
    $D_1$ & $f_1(\mathcal{D})$ & $f_2(\mathcal{D})$ & $D_1$ & $f_1(\mathcal{D})$ & $f_2(\mathcal{D})$ \\
    \hline 
    \textbf{38} & \textbf{0.004} & \textbf{0.103} & \textbf{38} & \textbf{0.004} & \textbf{0.103} \\
    \textbf{37} &  \textbf{0.009} & \textbf{0.097} & \textbf{37} &  \textbf{0.009} & \textbf{0.097}\\
    \textbf{35} &  \textbf{0.015} & \textbf{0.089} & \textbf{35} &  \textbf{0.015} & \textbf{0.089}  \\
    36 & 0.017 & 0.081  & - &  - & - \\
    33 & 0.023 & 0.071  & - & - & - \\
    32 & 0.038 & 0.068 & - & - & - \\
    30 & 0.042 & 0.058 & - &  - & - \\
    - & - & - & 29 &  0.058 & 0.068 \\
    28 & 0.067 & 0.051 & - & - & - \\
    - & - & - & 26 &  0.069 & 0.054 \\
    \textbf{24} &  \textbf{0.070} & \textbf{0.046} & \textbf{24} &  \textbf{0.070} & \textbf{0.046} \\
    \textbf{22} &  \textbf{0.088} & \textbf{0.041} & \textbf{22} &  \textbf{0.088} & \textbf{0.041} \\
    18 & 0.102 & 0.028 & - & - & - \\
    - & - & - & 15 &  0.134 & 0.033 \\
    \textbf{16} &  \textbf{0.142} & \textbf{0.026} & \textbf{16} &  \textbf{0.142} & \textbf{0.026} \\
    \hline
    \end{tabular}
    \label{tab:Pareto_front_comparison}
\end{table}
\spacingset{1.5}

We have validated the performance of MOBO-$\mathcal{D}^*$ by comparing it against the true Pareto front obtained as a result of full enumeration (exact solution) using all possible values of $D_1$, based on user-defined bounds (between 1 and 200). Although there is no notion of a reference point $\bm{r}$ in the enumeration, for fair comparison, we use the reference point calculated by MOBO-$\mathcal{D}^*$ to calculate the hyper-volume $(HV)$ for each enumerated point. This comparison is given in Figure \ref{fig:enumeration_vs_optimization_AM} (right) and shows that MOBO-$\mathcal{D}^*$ achieves an HV of 0.0617 which is just $2\%$ less than the maximum possible HV of 0.0632 in this case. 

\begin{figure}[H]
	\centering
    \vspace{15pt}
	\includegraphics[height=0.325\linewidth]{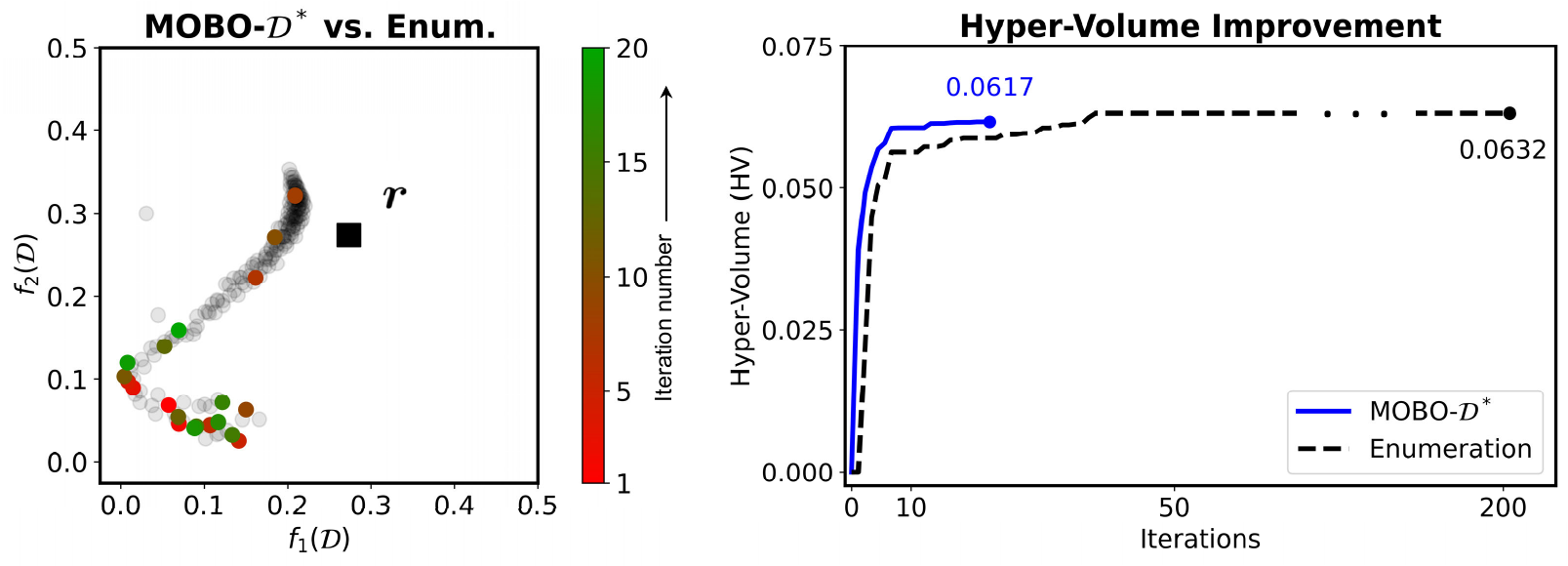}
    \spacingset{1}
    \caption{(left) MOBO-$\mathcal{D}^*$ iterations are plotted on a color scale from red to green (by iteration number), while enumeration calculations are plotted in black. (right) Comparison of $HVI$ between MOBO-$\mathcal{D}^*$ and enumeration. MOBO-$\mathcal{D}^*$ achieves an $HV$ which is just $2\%$ short of the maximum possible in 20 iterations.}
    \spacingset{1.5}
    \label{fig:enumeration_vs_optimization_AM}
\end{figure}

We have also evaluated the MOBO-$\mathcal {D}^*$ against two practical, reasonable approaches that can be used to approximate the Pareto front for this problem, as shown in Figure \ref{fig:hv_opt_vs_random}. The first approach randomly selects 20 points, the same number of iterations as has been used in the MOBO-$\mathcal{D}^*$, within bounds, and plots the objective values to approximate the Pareto front (instances 1, 2, and 3). The second approach is a structured approach in which we calculate objective values for 20 equidistant points, between 2 and 192, from within the bounds and plot them to approximate the Pareto front (instance 4). From Figure \ref{fig:hv_opt_vs_random}, it can be observed that MOBO-$\mathcal{D}^*$ gives $16.9\%$ higher $HVI$, on average, than the first approach, while $10.2\%$ higher $HVI$ than the second approach. 

\begin{figure}[ht]
	\centering
	\includegraphics[height=0.365\linewidth]{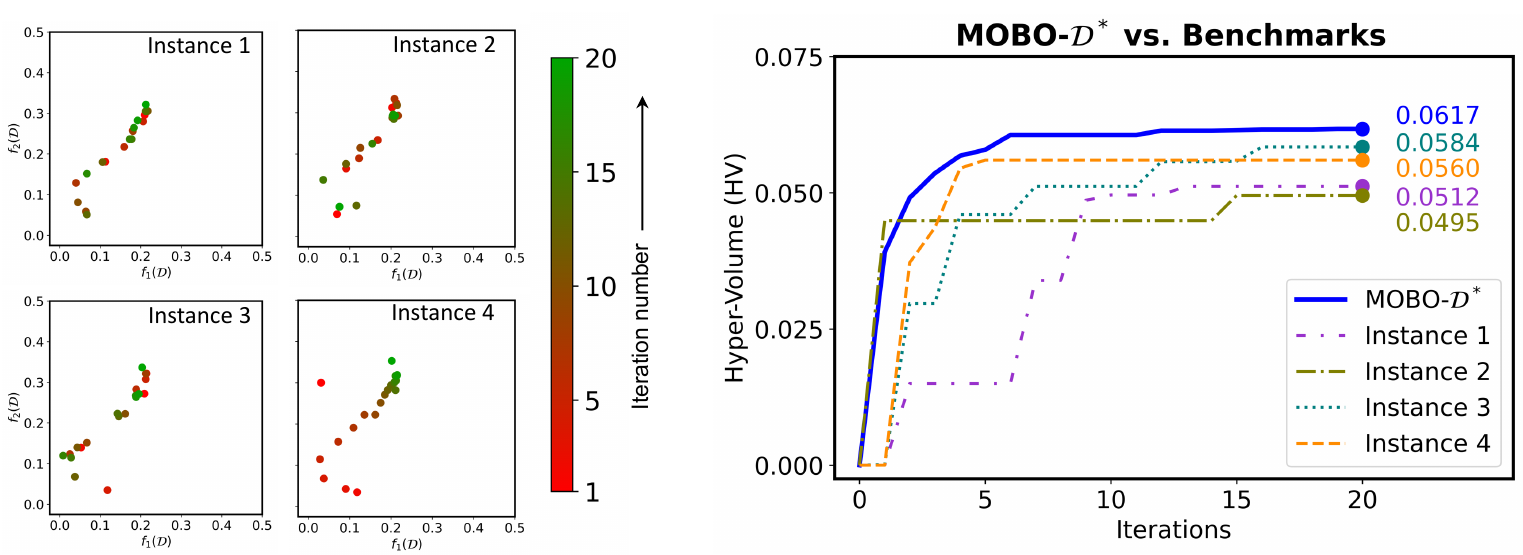}
    \spacingset{1}
    \vspace{-7pt}
    \caption{Comparison between MOBO-$\mathcal{D}^*$ and the four random instances. MOBO-$\mathcal{D}^*$ gives $16.9\%$ higher $HVI$, on average, than the random approach, while $10.2\%$ higher $HVI$ than the structured approach.}
    \vspace{-15pt}
    \spacingset{1.5}
    \label{fig:hv_opt_vs_random}
\end{figure}

Finally, we test the out-of-sample performance of the optimal solutions found. We hypothesize that if we use the same process conditions that were used to print this specimen to print more specimens, we can extend the analysis to the newly printed ones. Therefore, we use data from five more specimens with the same geometry and process conditions, and use the $\mathcal{D}$ values from the ND set to calculate both objective functions, i.e., $f_1(\mathcal{D})$ and $f_2(\mathcal{D})$, effectively giving the out-of-sample comparison for the solutions in the ND set. The results are given in Table \ref{tab:validation_benchmark_1D} where the mean value and standard deviation are given for these five specimens for $\mathcal{D}$ values. We notice the same trend, i.e., small $f_1({\mathcal{D}})$ values for small blocks and small $f_2(\mathcal{D})$ for large blocks, and vise versa for both objectives as in Figure \ref{fig:pareto_front_AM} (right). This shows that if the conditions that cause a certain phenomenon remain unchanged, we can use the optimized block size to extend the analysis to the new unseen data.

\spacingset{1}
\begin{table}[h]
	\centering
            \caption{Mean and standard deviation of $f_1({\mathcal{D}})$ and $f_2(\mathcal{D})$ values for the five specimens based on the potential choices of $\mathcal{D}^*$ from Figure \ref{fig:pareto_front_AM}.}
            \vspace{-5pt}
		\begin{tabular}{|C{0.5cm}| C{1cm} C{1.55cm} | C{1cm} C{1.55cm}|}
		\hline
            \multicolumn{1}{|c|}{$\mathcal{D}$} & \multicolumn{2}{c|}{$f_1(\mathcal{D})$} & \multicolumn{2}{c|}{$f_2(\mathcal{D})$} \\
            \hline
			$D_1$ &  Mean & Std. dev. & Mean & Std. dev. \\
                \hline
                38 & 0.055 & 0.040 & 0.029 & 0.002\\
                37 & 0.058 & 0.040 & 0.028 & 0.001\\
                35 & 0.056 & 0.040 & 0.027 & 0.001\\
                29 &  0.059 & 0.038 & 0.025 & 0.002\\
                26 &   0.060 & 0.040 & 0.021 & 0.001\\
                24 &  0.063 & 0.036 & 0.020 & 0.000\\
                22 & 0.064 & 0.035 & 0.019 & 0.002\\
                15 & 0.071 & 0.040 & 0.021 & 0.006\\
                16 & 0.070 & 0.035 & 0.020 & 0.004\\
                \hline
		\end{tabular}
		\label{tab:validation_benchmark_1D}
\end{table}
\spacingset{1.5}

\subsection{MOBO-$\mathcal{D^*}$ Rectangular block size optimization (2D)}\label{sec:case_study_AM_d*_two_param}
In this subsection, we optimize block size with rectangular blocks, i.e., $w_1 \neq w_2$, and have two block size parameters such that $\mathcal{D}=\{D_1, D_2\}$, hence making it a two-dimensional problem. Block size parameter $D_1$ is defined here as the number of segments into which the shorter dimension (1 $mm$) is divided, and $D_2$ is the number of segments into which the longer dimension (5 $mm$) is divided. Bounds are the same as we defined for the one-parameter problem, $1\leq D_1, D_2\leq U=200$. It is pertinent to highlight here that the computational cost of enumeration rises dramatically as possible combinations of $D_1$ and $D_2$ are $200 \times 200 = 40,000$. For stopping criteria, we choose $N=10$ and $\epsilon = 10^{-5}$ to incentivize our algorithm to explore more of the decision variable space, as it is significantly larger than the 1D problem. Based on this criteria, MOBO-$\mathcal{D}^*$ performs 67 iterations to approximate the Pareto front using the same approach as described in Section \ref{sec:case_study_AM_d*_one_param}. The optimization process and approximated Pareto front are shown in Figures \ref{fig:iterations_AM_two_param} and \ref{fig:pareto_front_AM_two_param}.

\begin{figure}[h]
	\centering
	\includegraphics[height=0.35\linewidth]{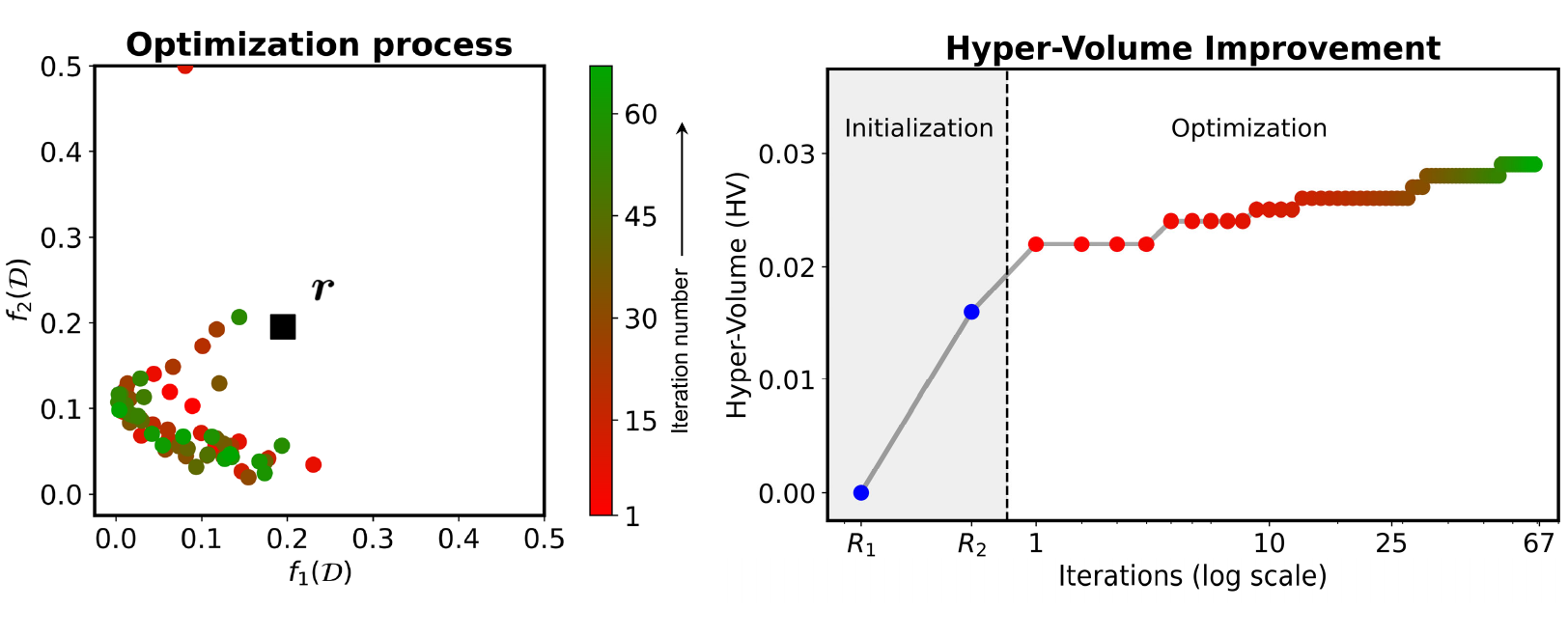}
    \spacingset{1}
    \caption{Optimization process consists of 67 iterations colored from red to green in ascending order.}
    \vspace{-10pt}
    \spacingset{1.5}
    \label{fig:iterations_AM_two_param}
\end{figure}

\begin{figure}[h]
	\centering
        {\includegraphics[height=0.35\linewidth]{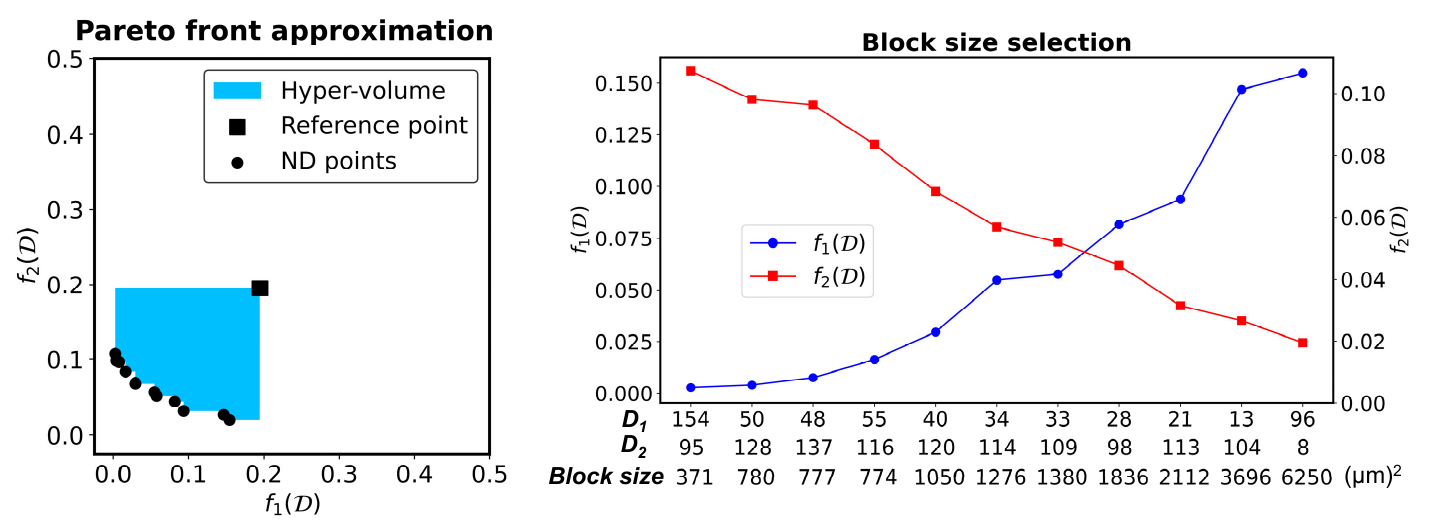}}
        \spacingset{1}
    \caption{(left) Approximated Pareto front based on the optimization process. (right) Both objective functions are plotted against the block size parameters $D_1$ and $D_2$ and show a trade-off between the two, which the decision maker can use to get the $\mathcal{D}^*$.}
    \spacingset{1.5}
\label{fig:pareto_front_AM_two_param}
\end{figure}

Then MOBO-$\mathcal{D}^*$ is evaluated against two benchmark approaches as was done for 1D case. Instances 1, 2, and 3 randomly select 67 points within bounds, while instance 4 uses 67 points located on a regular grid from the solution space in a structured way, as demonstrated in the \ref{app:structured_approach}. From Figure \ref{fig:hv_random_two_param}, it can be observed that MOBO-$\mathcal{D}^*$ gives $5.5 \%$ higher $HVI$, on average, than the random instances, while it gives $5.1 \%$ higher $HVI$ than the structured approach.

\begin{figure}[H]
	\centering
    \includegraphics[height=0.35\linewidth]{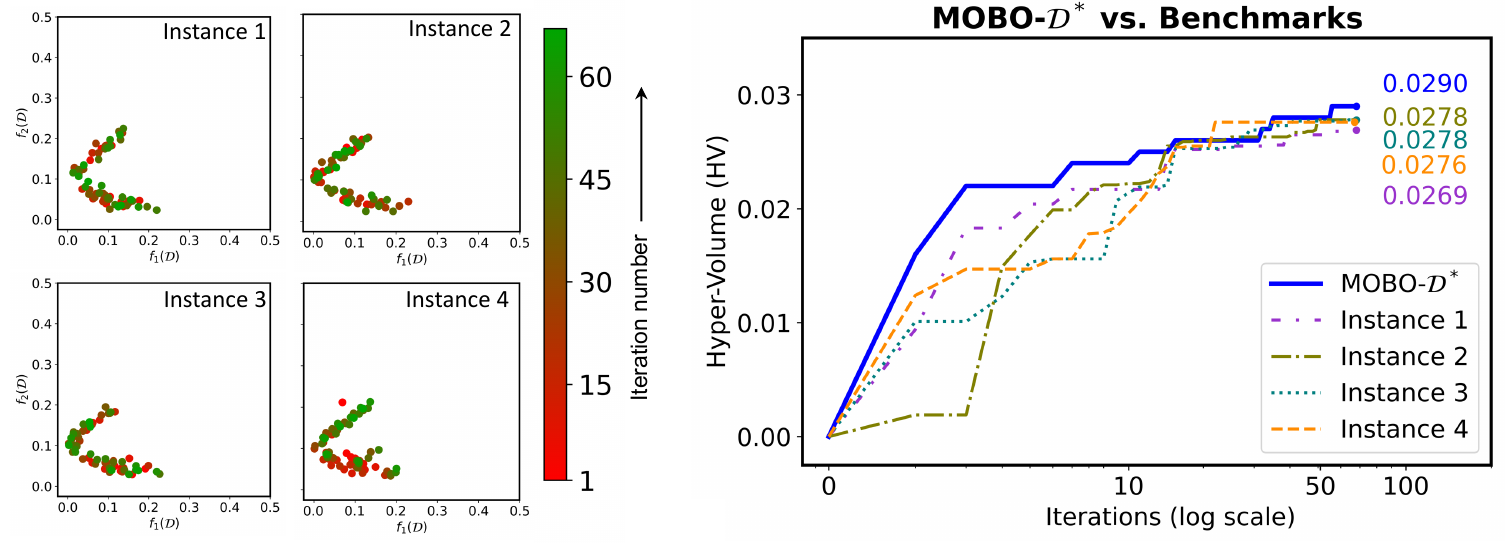}
    \spacingset{1}
    \caption{Comparison between MOBO-$\mathcal{D}^*$ and four random sets of $D_1$ and $D_2$ values. MOBO-$\mathcal{D}^*$ gives $5.5\%$ higher $HVI$, on average, than the random approach, while $5.1\%$ higher $HVI$ than the structured approach.}
    \spacingset{1.5}
    \label{fig:hv_random_two_param}
\end{figure}

Again, we validate our hypothesis of extending this analysis to new unseen data along the same lines as described in Section \ref{sec:case_study_AM_d*_one_param} by using data from the five new specimens to calculate both objective functions, i.e., $f_1(\mathcal{D})$ and $f_2(\mathcal{D})$.  The results are given in Table \ref{tab:validation_benchmark_2D}, where the mean value and standard deviation are given for these five specimens for $\mathcal{D}$ values. 

As already alluded to in Section \ref{sec:prob form}, we notice here that although the first set of block size parameters, i.e., $\mathcal{D}=\{157, 95\}$, is a minimizer of $f_1(\mathcal{D})$ in the training data set, here it gives the second highest value of $f_1(\mathcal{D})$, showing out-of-sample performance deterioration. This indicates possible overfitting to the training data and further underpins the importance of including goodness of fit as a second objective in the proposed framework. Apart from this, the overall trend remains the same for both objectives as in Figure \ref{fig:pareto_front_AM_two_param} (right), thus showing that we can use the optimized block size to extend the analysis to the new unseen data if the underlying process conditions remain the same.

\spacingset{1}
\begin{table}[H]
		\centering
            \caption{Mean and standard deviation of $f_1({\mathcal{D}})$ and $f_2(\mathcal{D})$ values for the five specimens based on the potential choices of $\mathcal{D}^*$ from Figure \ref{fig:pareto_front_AM_two_param}.}
		\begin{tabular}{|C{0.5cm} C{0.5cm}| C{1cm} C{1.55cm} | C{1cm} C{1.55cm}|}
		\hline
            \multicolumn{2}{|c|}{$\mathcal{D}$} & \multicolumn{2}{c|}{$f_1(\mathcal{D})$} & \multicolumn{2}{c|}{$f_2(\mathcal{D})$} \\
            \hline
			$D_1$ & $D_2$& Mean & Std. dev. & Mean & Std. dev. \\
                \hline
                154 & 95 & 0.077 & 0.034 & 0.030 & 0.002\\
                50 & 128 & 0.058 & 0.039 & 0.028 & 0.001\\
                48 & 137 & 0.056 & 0.040 & 0.028 & 0.003\\
                55 & 116 & 0.056 & 0.040 & 0.026 & 0.002\\
                40 & 120& 0.057 & 0.039 & 0.025 & 0.001\\
                34 & 114 & 0.057 & 0.040 & 0.023 & 0.002\\
                33 & 109 &  0.059& 0.042& 0.021& 0.001\\
                28 & 98 & 0.060 & 0.039 & 0.018 & 0.002\\
                21 & 113 & 0.063 & 0.037 & 0.018 & 0.001\\
                13 & 104 & 0.073 & 0.032 & 0.019 & 0.006\\
                96 & 8 & 0.088 & 0.047 & 0.022 & 0.005\\
                \hline
		\end{tabular}
		\label{tab:validation_benchmark_2D}
\end{table}
\spacingset{1.5}

\section{Extension to high dimensional domains}\label{sec:simulation}
\subsection{MOBO-$\mathcal{D^*}$ cube block size optimization (3D)}
As mentioned in Section \ref{sec:prob form}, MOBO-$\mathcal{D}^*$ can be applied to an $n$-dimensional domain; we demonstrate its generalizability using a simulated case study with a three-dimensional domain ($200 \times 200 \times 200$ $\mu m^3$) as shown in Figure \ref{fig:3D_simulation_lattice} (left). We assume that a phenomenon is observed at a resolution of 1 $mm$ in this domain, i.e., the distance between any two observed values is 1 $mm$ in all dimensions, and the total number of observations in this domain is $8,000,000$. Each observed value is an $i.i.d.$ sample from a normal distribution $\mathcal{N}(0,5^2)$. Using the same methodology as we adopted in Section \ref{sec:case studies}, we select a part of dimension $100 \times 100 \times 100$ $\mu m^3$ from the \textit{full domain} as shown in Figure \ref{fig:3D_simulation_lattice} (center). The aim here is to optimize block size, shown in Figure \ref{fig:3D_simulation_lattice} (right), for EVT. The \textit{selected domain} is used to build the model for the calculation of $\hat{q}$ to be used in $f_1(\mathcal{D})$ and calculation of $f_2(\mathcal{D})$. The maximum value in the entire \textit{full domain} is used as $q$ in the first objective $f_1(\mathcal{D})$. Each \textit{Block} has dimensions $w_1 \times w_2 \times w_3$ $\mu m^3$. Therefore, we have three block size parameters, i.e., $\mathcal{D} = \{D_1, D_2, D_3\}$, hence making it a three-dimensional problem. $D_j \; (j=1,2,3)$ is the number of parts in which $j^{th}$ dimension of the \textit{selected domain} is divided. Bounds on these parameters are set to $1\leq D_1, D_2, D_3\leq 50$. Thus, there are $125,000$, possible combinations of block size parameters and doing an enumeration would be computationally prohibitive. For stopping criteria, we choose $N=15$ and $\epsilon = 10^{-6}$, which achieved convergence at 75 iterations. Here, we used the window size, $N=15$ to allow our algorithm to explore more decision variable space as it is significantly larger in this case. 
\begin{figure}[H]
	\centering
	\includegraphics[height=0.33\linewidth]{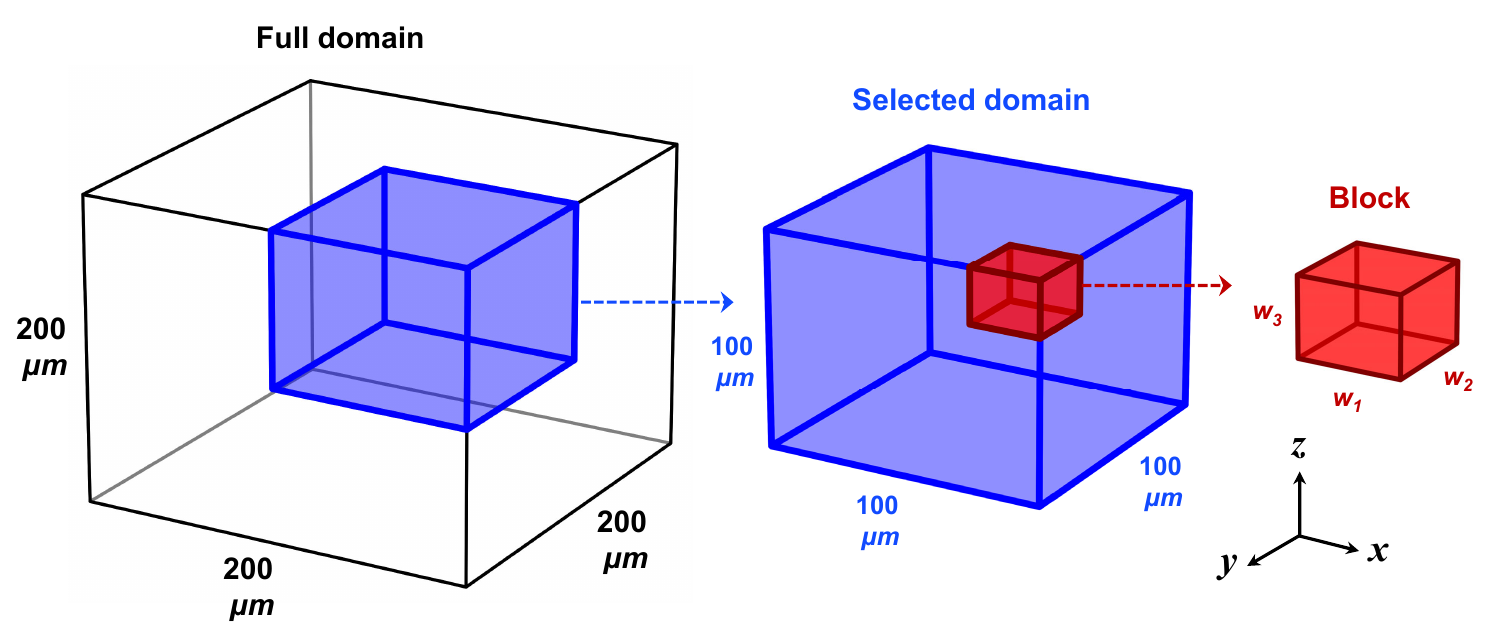}
    \spacingset{1}
    \caption{(left) Three-dimensional \textit{full domain} with dimensions $200 \times 200 \times 200$ $\mu m^3$ in which a phenomenon is observed at a resolution of 1 $mm$. (center) A \textit{selected domain} with dimensions $100 \times 100 \times 100$ $\mu m^3$ is used to build the model. (right) \textit{Block} which is used to extract the maxima values from the \textit{selected domain}. The aim here is to optimize its dimensions $w_1, w_2$, and $w_3$.}
    \spacingset{1.5}
    \label{fig:3D_simulation_lattice}
\end{figure}
\begin{figure}[h]
	\centering
	\includegraphics[height=0.35\linewidth]{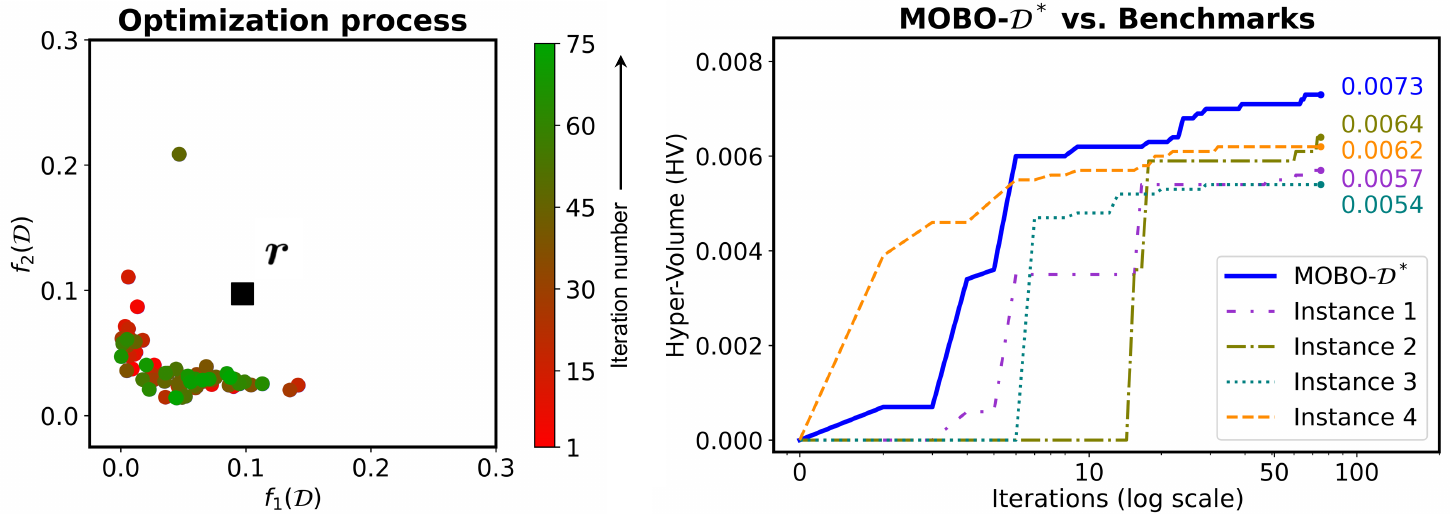}
    \spacingset{1}
    \caption{(left) Optimization process consists of 75 iterations colored from red to green color in ascending order. (right) Comparison of $HVI$ between MOBO-$\mathcal{D}^*$ and four random sets of $D_1, D_2$ and $D_3$ values. MOBO-$\mathcal{D}^*$ gives $25.8\%$ higher $HVI$, on average, than the random approach, while $17.7\%$ higher $HVI$ than the structured approach.}
    \spacingset{1.5}
    \label{fig:mobo_simulation}
\end{figure}

Figure \ref{fig:mobo_simulation} shows the iterations performed by MOBO-$\mathcal{D}^*$ and the comparison of MOBO-$\mathcal{D}^*$ with three random instances and one structured approach as was used in 2D case (demonstrated in the \ref{app:structured_approach}), in terms of hyper-volume ($HV$). MOBO-$\mathcal{D}^*$ gives $25.8\%$ higher $HVI$, on average, than the random approach, while $17.7\%$ higher $HVI$ than the structured approach. Note that in instance 4, initially, it outperforms MOBO-$\mathcal{D}^*$, which might have happened because it was exploring the entire solution space in a structured way, but then it stagnated. On the other hand, MOBO-$\mathcal{D}^*$ quickly catches up and exhibits a continuous improvement in $HV$. Figure \ref{fig:pareto_front_AM_three_param} gives the approximated Pareto front, which could be used for choosing  $\mathcal{D}^*$.

\begin{figure}[h]
	\centering
        {\includegraphics[height=0.325\linewidth]
        {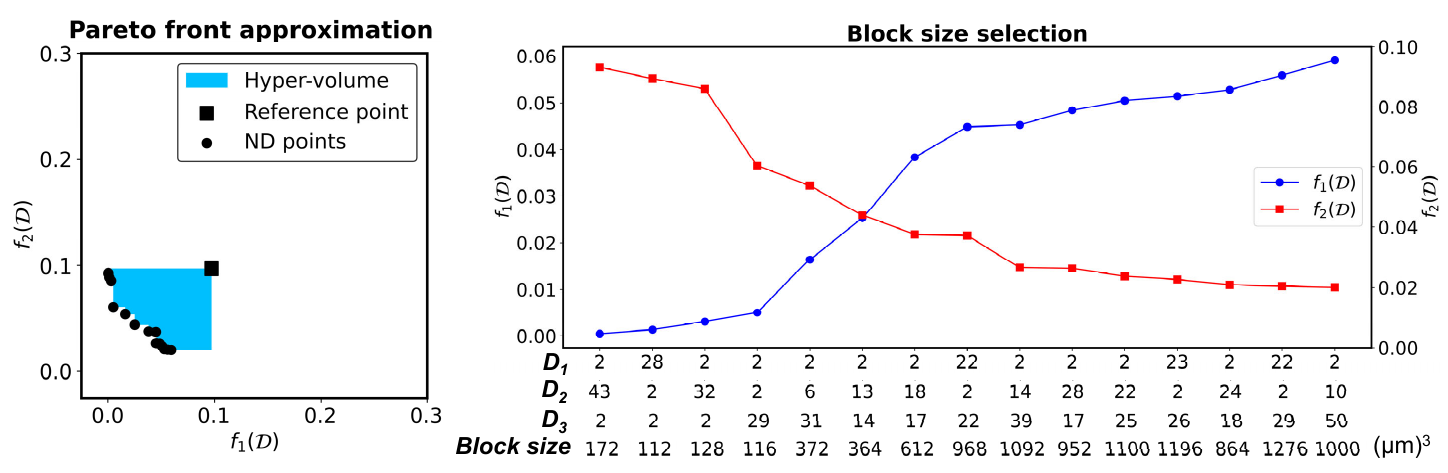}}
        \spacingset{1}
    \vspace{-10pt}
    \caption{(left) Approximated Pareto front based on the optimization process. (right) Both objective functions are plotted against the block size parameters $D_1$, $D_2$, and $D_3$ and show a trade-off between the two, which could be used to get the $\mathcal{D}^*$.}
    \spacingset{1.5}
    \label{fig:pareto_front_AM_three_param}
\end{figure}

To validate the hypothesis of extending the analysis to new realizations of the dataset, we simulate $500$ \textit{full domains} as shown in Figure \ref{fig:3D_simulation_lattice} (left) and use the $\mathcal{D}$ values to calculate both objective functions. The results are given in Table \ref{tab:validation_benchmark_3D} where the mean value and standard deviation are given for these $500$ simulated domains for the Pareto optimal $\mathcal{D}$ values. We notice the same trend for both objectives as in Figure \ref{fig:pareto_front_AM_three_param} (right), i.e., as the block size increases $f_1(\mathcal{D})$ also increases but $f_2(\mathcal{D})$ decreases and vise versa. This shows that if the \emph{conditions} that cause a certain phenomenon remain unchanged, the normal distribution $\mathcal{N}(0,5^2)$ in this case, we can use the optimized block size to extend the analysis to future realizations.

\spacingset{1}
\begin{table}[H]
		\centering
            \caption{Mean and standard deviation of $f_1({\mathcal{D}})$ and $f_2(\mathcal{D})$ values for 500 simulated datasets based on the potential choices of $\mathcal{D}^*$ from Figure \ref{fig:pareto_front_AM_three_param}.}
            \vspace{-5pt}
		\begin{tabular}{|C{0.5cm} C{0.5cm} C{0.5cm} | C{1cm} C{1.55cm} | C{1cm} C{1.55cm}|}
		\hline
            \multicolumn{3}{|c|}{$\mathcal{D}$} & \multicolumn{2}{c|}{$f_1(\mathcal{D})$} & \multicolumn{2}{c|}{$f_2(\mathcal{D})$} \\
            \hline
			$D_1$ & $D_2$ & $D_3$ & Mean & Std. dev. & Mean & Std. dev. \\
                \hline
                2 & 43 & 2 & 0.036 & 0.027 & 0.051 & 0.012\\
                28 & 2 & 2 & 0.035 & 0.028 & 0.061 & 0.014 \\
                2 & 32 & 2 & 0.035 & 0.027 & 0.058 & 0.013\\
                2 & 2 & 29 & 0.035 & 0.028 & 0.060 & 0.015\\
                2 & 6 & 31 & 0.043 & 0.029 & 0.038 & 0.009\\
                2 & 13 & 14 & 0.043 & 0.029 & 0.038 & 0.010\\
                2 & 18 & 17 & 0.047 & 0.030 & 0.032 & 0.008\\
                22 & 2 & 22 & 0.056 & 0.035 & 0.029 & 0.007\\
                2 & 14 & 39 & 0.059 & 0.037 & 0.028 & 0.006\\
                2 & 28 & 17 & 0.054 & 0.035 & 0.029 & 0.006\\
                2 & 22 & 25 & 0.061 & 0.037 & 0.028 & 0.006\\
                23 & 2 & 26 & 0.061 & 0.037 & 0.028 & 0.005\\
                2 & 24 & 18 & 0.056 & 0.035 & 0.029 & 0.006\\
                22 & 2 & 29 & 0.063 & 0.038 & 0.027 & 0.005\\
                2 & 10 & 50 & 0.062 & 0.037 & 0.028 & 0.006\\
                \hline
		\end{tabular}
        \vspace{5pt}
		\label{tab:validation_benchmark_3D}
\end{table}
\spacingset{1.5}

\subsection{Computational time comparison}\label{sec:time_comparison}
In this section, we do a computational time comparison between our proposed MOBO-$\mathcal{D}^*$ and the benchmark approaches. We have performed all of the above computations on a dual-core Intel(R) i7-10700 CPU @ 2.90GHz with 32 GB of RAM. We have plotted the cumulative CPU time taken for the iterations for each 1D, 2D, and 3D case,  as shown in Figure \ref{fig:time_comparison}. It can be clearly observed that the computational cost of our proposed MOBO-$\mathcal{D}^*$  is much lower than that of enumeration, while the $HVI$ given by MOBO-$\mathcal{D}^*$ is higher than the benchmark approaches, thus giving a good trade-off between solution quality and computational time. This dividend becomes more pronounced as the dimensions of the problem and search space become bigger.

\begin{figure}[H]
	\centering
	\includegraphics[width=0.8\linewidth]{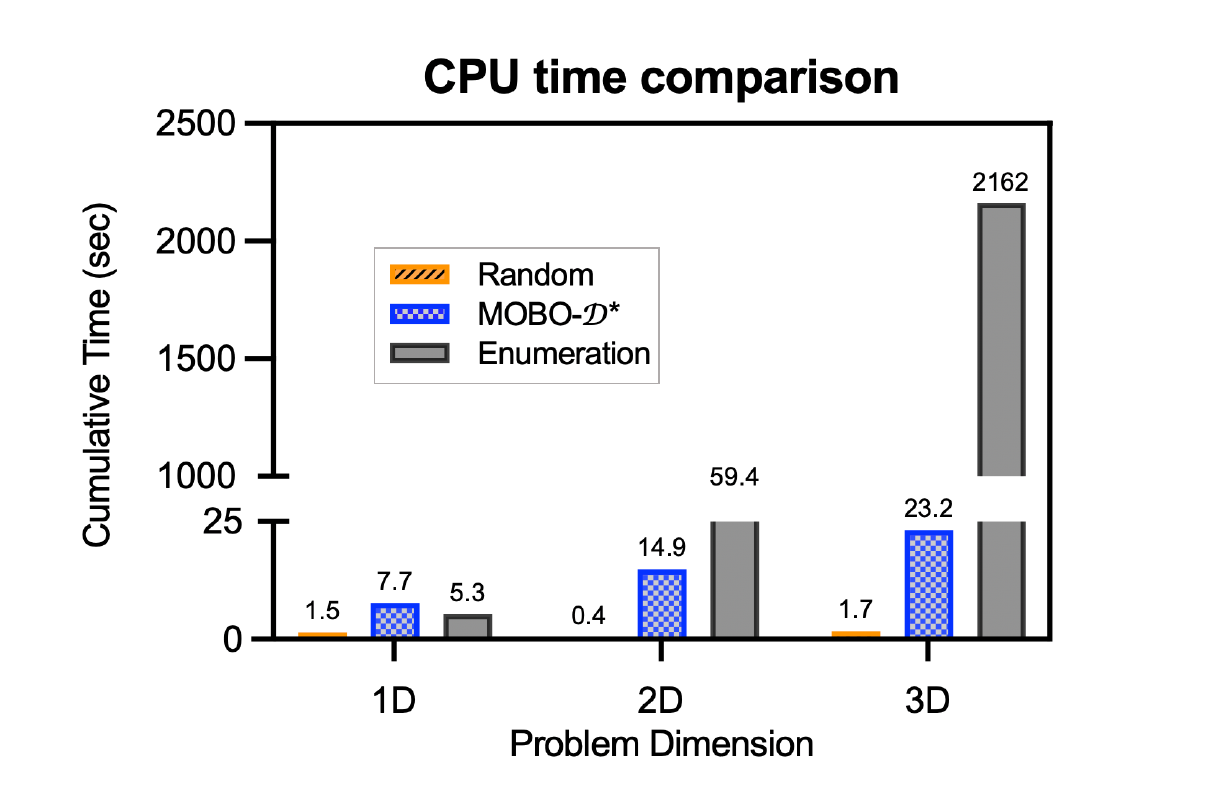}
   \spacingset{1}
  \caption{Computational time (seconds) comparison between the MOBO-$\mathcal{D}^*$ and the benchmark approaches.}
   \vspace{-10pt}
  \spacingset{1.5}
    \label{fig:time_comparison}
\end{figure}

\section{Conclusions and future work}\label{sec:conclusion}
Extreme value theory is extensively used in reliability analysis, and in this article, we have addressed the crucial question of deciding the block size in Block Maxima for extreme value analysis when no natural blocking parameters are available. We propose a general MOBO-$\mathcal{D}^*$ framework for multiple dimensions that can be used in different application domains. We have used the Bayesian framework to find the optimal block size for the domain from which maxima are extracted for the Gumbel distribution. The problem of selecting the optimal block size for maxima collection is formulated as a Multi-Objective Bayesian Optimization (MOBO) problem with two objectives: Kolmogorov-Smirnov (KS) test statistic to quantify goodness-of-fit and the prediction error. Both are minimized to balance bias (small blocks) and variance (large blocks). MOBO-$\mathcal{D}^*$ employs Gaussian Processes as surrogate models and maximizes the hyper-volume of the bi-objective space using a probabilistic acquisition function. A heuristic-based approach determines the appropriate placement of reference point for hyper-volume computation. The process iteratively refines the Pareto front, allowing decision-makers to select the optimal block size. The framework has been validated using both real-world additive manufacturing data and a synthetic simulation study. The proposed framework has also been scrutinized by comparing it against different benchmarks and has been shown to outperform them on average by $17\%$ (one-dimensional case), $5\%$ (two-dimensional case), and $22\%$ (three-dimensional case) in terms of hyper-volume improvement $(HVI)$. In terms of computational time, it could be very effective for the problems with huge solution space as demonstrated in Section \ref{sec:time_comparison}. Successful application of this framework shows that it can be effectively used to address the problem of finding optimal block sizes for accurate extreme value analysis.

There are several future avenues that could be explored based on the proposed research work. First, although this article has focused on the Gumbel distribution, one of the Generalized Extreme Value (GEV) family of distributions, further research efforts can be dedicated to developing a similar framework for Weibull and Fr\'{e}chet distributions. Second, the same methodology can be emulated to develop a tractable framework for finding the optimal threshold in the peak-over-threshold (POT) sampling methodology. Lastly, research efforts can be dedicated to extending the same methodology to develop a framework for multivariate Gumbel cupolas for extreme value analysis in high dimensions. 

\vspace{1em}
\hspace{-1.5em}\textbf{Data Availability Statement}: The data that support the findings of this study are available from the corresponding author upon reasonable request.

\vspace{1em}
\hspace{-1.5em}\textbf{Funding}: This work was supported by the Federal Aviation Administration [grant No. FAA-12-C-AM-AU-A5] and the National Science Foundation [grant No. CMMI-2239307]. 

\vspace{1em}
\hspace{-1.5em}\textbf{Disclosure of interest}: We have no conflicts of interest to disclose.

\spacingset{0.95}
\bibliographystyle{chicago}
\bibliography{main}

\begin{thebibliography}{}

\bibitem[\protect\citeauthoryear{Ahmad, Irfan, Maleki, Lee, Liu, Shao, and
  Shamsaei}{Ahmad et~al.}{2025}]{ahmad2025determining}
Ahmad, N., S.~Irfan, E.~Maleki, S.~Lee, J.~P. Liu, S.~Shao, and N.~Shamsaei
  (2025).
\newblock Determining critical surface features affecting fatigue behavior of
  additively manufactured {Ti-6Al-4V}.
\newblock {\em International Journal of Fatigue\/}~{\em 197}, 108956.

\bibitem[\protect\citeauthoryear{Asadi, Engelke, and Davison}{Asadi
  et~al.}{2018}]{asadi2018optimal}
Asadi, P., S.~Engelke, and A.~C. Davison (2018).
\newblock Optimal regionalization of extreme value distributions for flood
  estimation.
\newblock {\em Journal of Hydrology\/}~{\em 556}, 182--193.

\bibitem[\protect\citeauthoryear{Beirlant, Goegebeur, Segers, and
  Teugels}{Beirlant et~al.}{2006}]{beirlant2006statistics}
Beirlant, J., Y.~Goegebeur, J.~Segers, and J.~L. Teugels (2006).
\newblock {\em Statistics of extremes: theory and applications}.
\newblock John Wiley \& Sons.

\bibitem[\protect\citeauthoryear{B{\"u}cher and Segers}{B{\"u}cher and
  Segers}{2018}]{bucher2018inference}
B{\"u}cher, A. and J.~Segers (2018).
\newblock {Inference for heavy tailed stationary time series based on sliding
  blocks}.
\newblock {\em Electronic Journal of Statistics\/}~{\em 12\/}(1), 1098 -- 1125.

\bibitem[\protect\citeauthoryear{B{\"u}cher and Zhou}{B{\"u}cher and
  Zhou}{2021}]{bucher2021horse}
B{\"u}cher, A. and C.~Zhou (2021).
\newblock A horse race between the block maxima method and the
  peak--over--threshold approach.
\newblock {\em Statistical Science\/}~{\em 36\/}(3), 360--378.

\bibitem[\protect\citeauthoryear{Casella and Berger}{Casella and
  Berger}{2021}]{casella2021statistical}
Casella, G. and R.~L. Berger (2021).
\newblock {\em Statistical Inference}.
\newblock Cengage Learning.

\bibitem[\protect\citeauthoryear{Clifton, Tarassenko, McGrogan, King, King, and
  Anuzis}{Clifton et~al.}{2008}]{clifton2008bayesian}
Clifton, D.~A., L.~Tarassenko, N.~McGrogan, D.~King, S.~King, and P.~Anuzis
  (2008).
\newblock {Bayesian} extreme value statistics for novelty detection in
  gas-turbine engines.
\newblock In {\em 2008 IEEE Aerospace Conference}, pp.\  1--11.

\bibitem[\protect\citeauthoryear{Coles and Casson}{Coles and
  Casson}{1998}]{coles1998extreme}
Coles, S. and E.~Casson (1998).
\newblock Extreme value modelling of hurricane wind speeds.
\newblock {\em Structural Safety\/}~{\em 20\/}(3), 283--296.

\bibitem[\protect\citeauthoryear{Coles and Pericchi}{Coles and
  Pericchi}{2003}]{coles2003anticipating}
Coles, S. and L.~Pericchi (2003).
\newblock Anticipating catastrophes through extreme value modelling.
\newblock {\em Journal of the Royal Statistical Society Series C: Applied
  Statistics\/}~{\em 52\/}(4), 405--416.

\bibitem[\protect\citeauthoryear{Coles and Tawn}{Coles and
  Tawn}{1996}]{coles1996abayesian}
Coles, S.~G. and J.~A. Tawn (1996).
\newblock A {Bayesian} analysis of extreme rainfall data.
\newblock {\em Journal of the Royal Statistical Society Series C: Applied
  Statistics\/}~{\em 45\/}(4), 463--478.

\bibitem[\protect\citeauthoryear{Daulton, Balandat, and Bakshy}{Daulton
  et~al.}{2020}]{daulton2020differentiable}
Daulton, S., M.~Balandat, and E.~Bakshy (2020).
\newblock Differentiable expected hypervolume improvement for parallel
  multi-objective {Bayesian} optimization.
\newblock {\em Advances in Neural Information Processing Systems\/}~{\em 33},
  9851--9864.

\bibitem[\protect\citeauthoryear{Dkengne, Girard, and Ahiad}{Dkengne
  et~al.}{2020}]{Dkengne2020}
Dkengne, P.~S., S.~Girard, and S.~Ahiad (2020).
\newblock {An automatic procedure to select a block size in the continuous
  generalized extreme value model estimation}.
\newblock Working paper. National Institute for Research in Digital Science and
  Technology (INRIA), France.

\bibitem[\protect\citeauthoryear{Elkahlout}{Elkahlout}{2006}]{elkahlout2006bayes}
Elkahlout, G. (2006).
\newblock Bayes estimators for the extreme-value reliability function.
\newblock {\em Computers \& Mathematics with Applications\/}~{\em 51\/}(3-4),
  673--679.

\bibitem[\protect\citeauthoryear{Engeland, Hisdal, and Frigessi}{Engeland
  et~al.}{2004}]{engeland2004practical}
Engeland, K., H.~Hisdal, and A.~Frigessi (2004).
\newblock Practical extreme value modelling of hydrological floods and
  droughts: a case study.
\newblock {\em Extremes\/}~{\em 7}, 5--30.

\bibitem[\protect\citeauthoryear{Ferreira and De~Haan}{Ferreira and
  De~Haan}{2015}]{ferreira2015block}
Ferreira, A. and L.~De~Haan (2015).
\newblock On the {Block Maxima} method in extreme value theory: {PWM}
  estimators.
\newblock {\em The Annals of statistics\/}, 276--298.

\bibitem[\protect\citeauthoryear{Fisher and Tippett}{Fisher and
  Tippett}{1928}]{fisher1928limiting}
Fisher, R.~A. and L.~H.~C. Tippett (1928).
\newblock Limiting forms of the frequency distribution of the largest or
  smallest member of a sample.
\newblock In {\em Mathematical proceedings of the Cambridge philosophical
  society}, Volume~24, pp.\  180--190.

\bibitem[\protect\citeauthoryear{Fox and Pintar}{Fox and
  Pintar}{2021}]{fox2021prediction}
Fox, J.~C. and A.~L. Pintar (2021).
\newblock Prediction of extreme value areal parameters in laser powder bed
  fusion of nickel superalloy 625.
\newblock {\em Surface Topography: Metrology and Properties\/}~{\em 9\/}(2),
  025033.

\bibitem[\protect\citeauthoryear{Gnanasambandam, Shen, Law, Dou, and
  Kong}{Gnanasambandam et~al.}{2024}]{gnanasambandam2024deep}
Gnanasambandam, R., B.~Shen, A.~C.~C. Law, C.~Dou, and Z.~Kong (2024).
\newblock Deep {Gaussian} process for enhanced {Bayesian} optimization and its
  application in additive manufacturing.
\newblock {\em IISE Transactions\/}, 1--14.

\bibitem[\protect\citeauthoryear{Gockel, Sheridan, Koerper, and Whip}{Gockel
  et~al.}{2019}]{gockel2019influence}
Gockel, J., L.~Sheridan, B.~Koerper, and B.~Whip (2019).
\newblock The influence of additive manufacturing processing parameters on
  surface roughness and fatigue life.
\newblock {\em International Journal of Fatigue\/}~{\em 124}, 380--388.

\bibitem[\protect\citeauthoryear{Gumbel}{Gumbel}{1958}]{gumbel1958statistics}
Gumbel, E.~J. (1958).
\newblock {\em Statistics of extremes}.
\newblock Columbia university press.

\bibitem[\protect\citeauthoryear{Gurung, Sarkar, Singh, and Lama}{Gurung
  et~al.}{2021}]{gurung2021modelling}
Gurung, B., K.~P. Sarkar, K.~Singh, and A.~Lama (2021).
\newblock Modelling annual maximum temperature of {India}: a distributional
  approach.
\newblock {\em Theoretical and Applied Climatology\/}~{\em 145\/}(3-4),
  979--988.

\bibitem[\protect\citeauthoryear{Haan and Ferreira}{Haan and
  Ferreira}{2006}]{haan2006extreme}
Haan, L. and A.~Ferreira (2006).
\newblock {\em Extreme Value Theory: An Introduction}, Volume~3.
\newblock Springer.

\bibitem[\protect\citeauthoryear{Haviland}{Haviland}{1964}]{haviland1964engineering}
Haviland, R. (1964).
\newblock Engineering applications of extreme value theory.
\newblock Technical report, SAE Technical Paper.

\bibitem[\protect\citeauthoryear{Kang, Ko, and Huh}{Kang
  et~al.}{2015}]{kang2015determination}
Kang, D., K.~Ko, and J.~Huh (2015).
\newblock Determination of extreme wind values using the {Gumbel} distribution.
\newblock {\em Energy\/}~{\em 86}, 51--58.

\bibitem[\protect\citeauthoryear{Kass and Wasserman}{Kass and
  Wasserman}{1996}]{kass1996selection}
Kass, R.~E. and L.~Wasserman (1996).
\newblock The selection of prior distributions by formal rules.
\newblock {\em Journal of the American statistical Association\/}~{\em
  91\/}(435), 1343--1370.

\bibitem[\protect\citeauthoryear{Lee, Rasoolian, Silva, Pegues, and
  Shamsaei}{Lee et~al.}{2021}]{lee2021surface}
Lee, S., B.~Rasoolian, D.~F. Silva, J.~W. Pegues, and N.~Shamsaei (2021).
\newblock Surface roughness parameter and modeling for fatigue behavior of
  additive manufactured parts: A non-destructive data-driven approach.
\newblock {\em Additive Manufacturing\/}~{\em 46}, 102094.

\bibitem[\protect\citeauthoryear{Li, Poudel, Shao, Shamsaei, and Liu}{Li
  et~al.}{2024}]{li2024nondestructive}
Li, A., A.~Poudel, S.~Shao, N.~Shamsaei, and J.~Liu (2024).
\newblock Nondestructive fatigue life prediction for additively manufactured
  metal parts through a multimodal transfer learning framework.
\newblock {\em IISE Transactions\/}, 1--16.

\bibitem[\protect\citeauthoryear{Li and Zhang}{Li and
  Zhang}{2014}]{li2014extreme}
Li, Z. and Y.~Zhang (2014).
\newblock Extreme value theory-based structural health prognosis method using
  reduced sensor data.
\newblock {\em Structure and Infrastructure Engineering\/}~{\em 10\/}(8),
  988--997.

\bibitem[\protect\citeauthoryear{Melchers}{Melchers}{2021}]{melchers2021new}
Melchers, R.~E. (2021).
\newblock New insights from probabilistic modelling of corrosion in structural
  reliability analysis.
\newblock {\em Structural Safety\/}~{\em 88}, 102034.

\bibitem[\protect\citeauthoryear{Moins, Arbel, Girard, and Dutfoy}{Moins
  et~al.}{2023}]{moins2023reparameterization}
Moins, T., J.~Arbel, S.~Girard, and A.~Dutfoy (2023).
\newblock Reparameterization of extreme value framework for improved {Bayesian}
  workflow.
\newblock {\em Computational Statistics \& Data Analysis\/}~{\em 187}, 107807.

\bibitem[\protect\citeauthoryear{Mousa, Jaheen, and Ahmad}{Mousa
  et~al.}{2002}]{mousa2002bayesian}
Mousa, M.~A., Z.~Jaheen, and A.~Ahmad (2002).
\newblock Bayesian estimation, prediction and characterization for the {Gumbel}
  model based on records.
\newblock {\em Statistics: A Journal of Theoretical and Applied
  Statistics\/}~{\em 36\/}(1), 65--74.

\bibitem[\protect\citeauthoryear{Murakami, Masuo, Tanaka, and
  Nakatani}{Murakami et~al.}{2019}]{Murakami_2019}
Murakami, Y., H.~Masuo, Y.~Tanaka, and M.~Nakatani (2019).
\newblock Defect analysis for additively manufactured materials in fatigue from
  the viewpoint of quality control and statistics of extremes.
\newblock {\em Procedia Structural Integrity\/}~{\em 19}, 113--122.

\bibitem[\protect\citeauthoryear{Najibi and Devineni}{Najibi and
  Devineni}{2018}]{najibi2018recent}
Najibi, N. and N.~Devineni (2018).
\newblock Recent trends in the frequency and duration of global floods.
\newblock {\em Earth System Dynamics\/}~{\em 9\/}(2), 757--783.

\bibitem[\protect\citeauthoryear{Niemann and Diburg}{Niemann and
  Diburg}{2013}]{niemann2013statistics}
Niemann, H.-J. and S.~Diburg (2013).
\newblock Statistics of extreme climatic actions based on the {Gumbel}
  probability distributions with an upper limit.
\newblock {\em Computers \& Structures\/}~{\em 126}, 193--198.

\bibitem[\protect\citeauthoryear{Nikfar, Irfan, Baugh, Mahmood, Ahmad, Liu,
  Jackson, Schulze, Shao, Silva, et~al.}{Nikfar
  et~al.}{2025}]{nikfar2025extreme}
Nikfar, M., S.~Irfan, L.~Baugh, S.~Mahmood, N.~Ahmad, J.~Liu, R.~L. Jackson,
  K.~Schulze, S.~Shao, D.~F. Silva, et~al. (2025).
\newblock On extreme value theory-based estimation of surface quality for metal
  additive manufacturing.
\newblock {\em Progress in Additive Manufacturing\/}, 1--16.

\bibitem[\protect\citeauthoryear{Ouyang, Zhu, Ye, Park, and Wang}{Ouyang
  et~al.}{2022}]{ouyang2022robust}
Ouyang, L., S.~Zhu, K.~Ye, C.~Park, and M.~Wang (2022).
\newblock Robust {Bayesian} hierarchical modeling and inference using scale
  mixtures of normal distributions.
\newblock {\em IISE Transactions\/}~{\em 54\/}(7), 659--671.

\bibitem[\protect\citeauthoryear{{\"O}zari, Eren, and Saygin}{{\"O}zari
  et~al.}{2019}]{ozari2019}
{\"O}zari, {\c{C}}., {\"O}.~Eren, and H.~Saygin (2019).
\newblock A new methodology for the {Block Maxima} approach in selecting the
  optimal block size.
\newblock {\em Tehni{\v{c}}ki Vjesnik\/}~{\em 26\/}(5), 1292--1296.

\bibitem[\protect\citeauthoryear{Pickands~III}{Pickands~III}{1975}]{pickands1975statistical}
Pickands~III, J. (1975).
\newblock Statistical inference using extreme order statistics.
\newblock {\em the Annals of Statistics\/}, 119--131.

\bibitem[\protect\citeauthoryear{Purohit and Lalit}{Purohit and
  Lalit}{2022}]{purohit2022european}
Purohit, S.~U. and P.~N. Lalit (2022).
\newblock European option pricing using {Gumbel} distribution.
\newblock {\em International Journal of Financial Engineering\/}~{\em 9\/}(1),
  2141002.

\bibitem[\protect\citeauthoryear{Roussel, Hanuka, and Edelen}{Roussel
  et~al.}{2021}]{roussel2021multiobjective}
Roussel, R., A.~Hanuka, and A.~Edelen (2021).
\newblock Multiobjective {Bayesian} optimization for online accelerator tuning.
\newblock {\em Physical Review Accelerators and Beams\/}~{\em 24\/}(6), 062801.

\bibitem[\protect\citeauthoryear{Sanaei and Fatemi}{Sanaei and
  Fatemi}{2021}]{sanaei2021defects}
Sanaei, N. and A.~Fatemi (2021).
\newblock Defects in additive manufactured metals and their effect on fatigue
  performance: a state-of-the-art review.
\newblock {\em Progress in Materials Science\/}~{\em 117}, 100724.

\bibitem[\protect\citeauthoryear{Shen, Mickley, and Gilleland}{Shen
  et~al.}{2016}]{shen2016impact}
Shen, L., L.~J. Mickley, and E.~Gilleland (2016).
\newblock Impact of increasing heat waves on us ozone episodes in the 2050s:
  Results from a multimodel analysis using extreme value theory.
\newblock {\em Geophysical Research Letters\/}~{\em 43\/}(8), 4017--4025.

\bibitem[\protect\citeauthoryear{Vidal}{Vidal}{2014}]{vidal2014bayesian}
Vidal, I. (2014).
\newblock A {Bayesian} analysis of the {Gumbel} distribution: an application to
  extreme rainfall data.
\newblock {\em Stochastic Environmental Research and Risk Assessment\/}~{\em
  28}, 571--582.

\bibitem[\protect\citeauthoryear{Wang, You, Wu, Zhang, and Bin}{Wang
  et~al.}{2016}]{wang2018}
Wang, J., S.~You, Y.~Wu, Y.~Zhang, and S.~Bin (2016).
\newblock A method of selecting the {Block Size} of {BMM} for estimating
  extreme loads in engineering vehicles.
\newblock {\em Mathematical Problems in Engineering\/}~{\em 2016}, 6372197.

\bibitem[\protect\citeauthoryear{Yang, Emmerich, Deutz, and B{\"a}ck}{Yang
  et~al.}{2019}]{yang2019multi}
Yang, K., M.~Emmerich, A.~Deutz, and T.~B{\"a}ck (2019).
\newblock Multi-objective {Bayesian} global optimization using expected
  hypervolume improvement gradient.
\newblock {\em Swarm and Evolutionary Computation\/}~{\em 44}, 945--956.

\bibitem[\protect\citeauthoryear{Zhao, Liang, Parlikad, and Xie}{Zhao
  et~al.}{2022}]{zhao2022performance}
Zhao, X., Z.~Liang, A.~K. Parlikad, and M.~Xie (2022).
\newblock Performance-oriented risk evaluation and maintenance for multi-asset
  systems: A {Bayesian} perspective.
\newblock {\em IISE Transactions\/}~{\em 54\/}(3), 251--270.

\bibitem[\protect\citeauthoryear{Zou, Volgushev, and B{\"u}cher}{Zou
  et~al.}{2021}]{zou2021multiple}
Zou, N., S.~Volgushev, and A.~B{\"u}cher (2021).
\newblock {Multiple block sizes and overlapping blocks for multivariate time
  series extremes}.
\newblock {\em The Annals of Statistics\/}~{\em 49\/}(1), 295 -- 320.

\end{thebibliography}

\appendix
\renewcommand\thesection{Appendix \Alph{section}}
\spacingset{1.5}
\section{GEV parameter estimation}\label{app:parameter_estimation}
\textbf{MLE estimates.} Likelihood function $\mathcal{L}(\mathbf{\Theta}|x)$ for the Gumbel distribution can be written as follows:
\begin{align}\label{eq:likelihood}
\mathcal{L}(\mathbf{\Theta}|x) =& \prod_{i=1}^{m(\mathcal{D})} \frac{1}{\sigma} \exp \left\{-\frac{x_i-\mu}{\sigma}\right\} \exp\left\{-\exp \left\{ \left( -\frac{x_i-\mu}{\sigma}\right) \right\}\right\}, \\
	\mathcal{L}(\mathbf{\Theta}|x)=&  \frac{1}{\sigma^{m(\mathcal{D})}} \exp \left\{\sum\limits_{i=1}^{m(\mathcal{D})}- \frac{x_i - \mu}{\sigma} - \sum\limits_{i=1}^{m(\mathcal{D})} \exp \left\{- \frac{x_i - \mu}{\sigma}\right\}\right\},
\end{align}
and the MLE estimates of the parameters can be obtained using following equation:
\begin{align}\label{eq:map_estimates}
	\mathbf{\widehat{\Theta}}_{MLE} = \underset{\mathbf{\Theta}}{\text{argmax}} \hspace{0.5em}	   \frac{1}{\sigma^{m(\mathcal{D})}} \exp \left\{\sum\limits_{i=1}^{m(\mathcal{D})}- \frac{x_i - \mu}{\sigma} - \sum\limits_{i=1}^{m(\mathcal{D})} \exp \left\{- \frac{x_i - \mu}{\sigma}\right\}\right\}.
\end{align}

\hspace{-1.5em}\textbf{MAP estimates.} Bayesian inference of the parameters has a number of potential advantages over the point estimators. First, it is not dependent on regularity assumptions \citep{beirlant2006statistics, casella2021statistical}; hence, it offers a viable alternative when point estimators may break down. Second, Bayesian inference is robust against overfitting the estimated parameter because it produces a full posterior distribution, which encapsulates the associated uncertainty. Third, it provides the ability to estimate functions of parameters, which we can not get in the case of point estimators. Using Bayes theorem:
\begin{align}\label{eq:joint_post_parameters}
	h(\mathbf{\Theta|x}) \propto \mathcal{L}(\mathbf{\Theta|x}) \times  \pi(\mathbf{\Theta}),
\end{align}

where $\mathcal{L}(\mathbf{\Theta}|x)$ is the likelihood function as given above and $\pi(\mathbf{\Theta})$ is the prior assigned to the parameters $\mathbf{\Theta}$. We have used Jeffreys Prior, which is a popular choice among the \textit{objective} priors \citep{kass1996selection, ouyang2022robust} due to its immunity to the analyst's biases. 
Another choice could be to use a \textit{subjective} prior, which uses the already existing knowledge to be encoded in the form of a well-known statistical distribution. Many different types and combinations of subjective priors have been used for the Gumbel parameter estimation \citep{vidal2014bayesian, mousa2002bayesian,coles1996abayesian, coles2003anticipating}. 
However, subjective priors could lead to encoding `too much' information, thus leading to bias. Therefore, an objective prior helps address this concern. By definition, Jeffreys Prior can be given as:
\begin{align}\label{eq:def_jeffreys}
		\pi_{J} (\mathbf{\Theta}) = \biggl[\text{det}\Bigl(\mathrm{I} (\mathbf{\Theta})\Bigr)\biggr]^{1/2} = & \Biggl[\text{det}\begin{pmatrix}
\mathrm{I}_{11} & \mathrm{I}_{12} \\
\mathrm{I}_{21} & \mathrm{I}_{22}
\end{pmatrix}\Biggr]^{1/2} ,
\end{align}
where $\mathrm{I(\mathbf{\Theta})}$ is the Fisher's Information matrix whose individual elements are given by:
\begin{align}\label{eq:fisher_info_matrix}
	\mathrm{I}_{ij} = \mathbb{E} \Biggl[-\frac{\partial^2 \log g(x| \theta)}{\partial \theta_i \theta_j}\Biggr].
\end{align}
Jeffreys prior for the Gumbel distribution (\cite{elkahlout2006bayes}) is given as:
\begin{align}\label{eq:jeffreys_gumbel}
		\pi_{J} (\mathbf{\Theta}) \propto \frac{1}{\sigma ^2}.
\end{align}
Using (\ref{eq:jeffreys_gumbel}) it can be shown that it is an improper prior because $\int_{0}^{+\infty} \int_{-\infty}^{+\infty} \pi_{J} (\mathbf{\Theta}) d\mu d\sigma$ does not converge over the support of parameters ($\mu, \sigma$); hence the convergence of the posterior will have to be shown for justifying its use. Using (\ref{eq:likelihood}), (\ref{eq:joint_post_parameters}) and (\ref{eq:jeffreys_gumbel}), the joint posterior of the parameters can be given as:
\begin{align}\label{eq:joint_posterior}
	h(\mathbf{\Theta}|x) \propto   \frac{1}{\sigma^{m(\mathcal{D})+2}} \exp \left\{\sum\limits_{i=1}^{m(\mathcal{D})}- \frac{x_i - \mu}{\sigma} - \sum\limits_{i=1}^{m(\mathcal{D})} \exp \left\{- \frac{x_i - \mu}{\sigma}\right\}\right\}.
\end{align}
\begin{proposition}\label{prop:2_posterior_convergence}
	Joint posterior of the parameters $(\mu, \:\sigma)$ given by (\ref{eq:joint_posterior}) is a proper posterior distribution that converges over the support of parameters.
\end{proposition}
\begin{proof}
	by induction. For $m(\mathcal{D})=1$:
	\begin{align}
		C_1 = &  K \int_{z}^{+\infty} \int_{-\infty}^{+\infty} h(\mu,\sigma|x) d\mu \: d\sigma, \quad z>0, \nonumber \\
		= & K   \int_{z}^{+\infty}\frac{1}{\sigma^{3}} \int_{-\infty}^{+\infty}  \exp \left\{- \frac{x - \mu}{\sigma} -  \exp \left\{- \frac{x - \mu}{\sigma}\right\}\right\} d\mu \: d\sigma. \nonumber
	\end{align}
	Letting $x-\mu = a$ and $\sigma = b$, same integral can be written with new variables as following:
	\begin{align}
		C_1 = &  \; K \int_{z}^{+\infty} \int_{- \infty}^{+\infty} \frac{1}{b^3} \exp \left\{-\frac{a}{b} - \exp \left\{- \frac{a}{b}\right\} \right\} (da) (db), \nonumber \\
		= & \; K  \int_{z}^{+\infty}  \frac{1}{b^2} \: db,  \nonumber \\
		= & \; K \: \frac{1}{z}, \quad \text{if} \quad z>0 \quad \Rightarrow \; C_1 < \infty \quad \Rightarrow \; C_n < \infty \quad \forall \;\; n. \nonumber
	\end{align}
	Therefore, the joint posterior distribution of the parameters is a proper distribution.
\end{proof}
For point estimates, which will be used for the calculation of first and second objective values, maximum a posteriori (MAP) can be used as follows:
\begin{align}\label{eq:map_estimates}
	\mathbf{\widehat{\Theta}}_{MAP} = \underset{\mathbf{\Theta}}{\text{argmax}} \hspace{0.5em}	h(\mathbf{\Theta}|x) \propto   \frac{1}{\sigma^{m(\mathcal{D})+2}} \exp \left\{\sum\limits_{i=1}^{m(\mathcal{D})}- \frac{x_i - \mu}{\sigma} - \sum\limits_{i=1}^{m(\mathcal{D})} \exp \left\{- \frac{x_i - \mu}{\sigma}\right\}\right\}.
\end{align}

Although we choose to use the Bayesian method in our proposed algorithm for the aforementioned reasons, our methodology does not depend on it exclusively and just needs the point estimates of Gumbel parameters for the calculation of both objectives. Once we have the point estimates $\mathbf{\widehat{\Theta}}$, (\ref{eq:return_level}) can be used to estimate the most extreme value for any number of blocks $m(\mathcal{D})$, and $G(x)$ can be used in the objective functions.

\section{Structured approach for selecting the block sizes}\label{app:structured_approach}

\begin{figure}[H]
	\centering
	\includegraphics[width=0.65\linewidth]{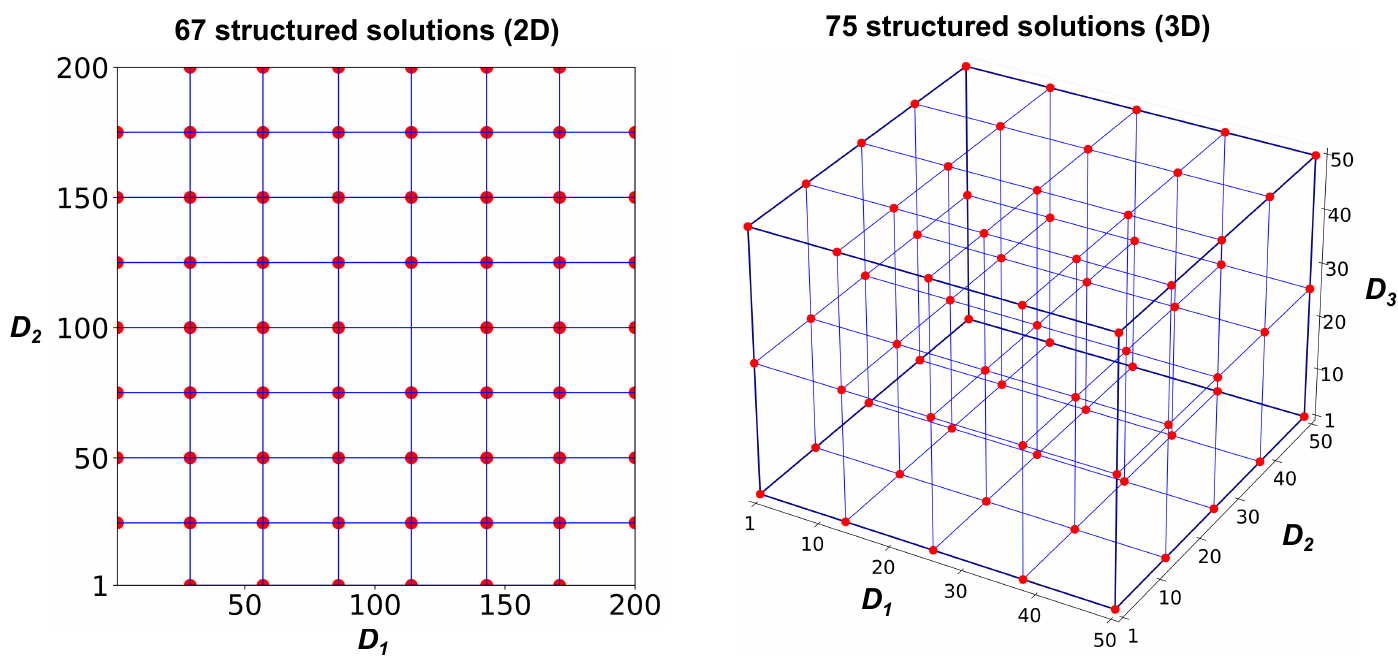}
   \spacingset{1}
  \caption{Structured approach used for generating instance 4 for benchmarking 2D and 3D cases. Out of the total possible 40,000 ($200 \times 200$) solutions for the 2D case and 125,000 ($50 \times 50 \times 50$) solutions for the 3D case, we use the 67 highlighted by red dots (left) and 75 highlighted by red dots (right), respectively. This approach covers the entire solution space in a structured way. In the 2D case, to match the number of iterations by MOBO-$\mathcal{D}^*$, i.e., 67, we skip 5 points, 4 corner points, and 1 in the center, on the $8 \times 9$ grid.} 
   \vspace{-5pt}
  \spacingset{1.5}
    \label{fig:lattice_cube}
\end{figure}

\end{document}